\documentclass[letterpaper,11pt]{amsart}
\pdfoutput=1

\usepackage{color}
\definecolor{darkblue}{rgb}{0.,0.,0.4}
\definecolor{darkred}{rgb}{0.5,0.,0.}
\usepackage[pdftex,colorlinks=true,linkcolor=blue,citecolor=darkred,urlcolor=darkblue]{hyperref}

\usepackage{amsmath,amsthm,amsfonts,amssymb,braket}
\usepackage[all]{xy}
\usepackage{fullpage}
\usepackage{mathtools}
\usepackage[numbers]{natbib}

\newtheorem{lemma}{Lemma}[section]
\newtheorem{theorem}[lemma]{Theorem}

\newtheorem{corollary}[lemma]{Corollary}

\newtheorem{proposition}[lemma]{Proposition}
\newtheorem{proposition-definition}[lemma]{Proposition-Definition}
\theoremstyle{definition}
\newtheorem{definition}[lemma]{Definition}
\newtheorem{remark}[lemma]{Remark}
\newtheorem{example}[lemma]{Example}

\newcommand{\CC}{{\mathbb{C}}}
\newcommand{\RR}{{\mathbb{R}}}
\newcommand{\ZZ}{{\mathbb{Z}}}

\newcommand{\FF}{{\mathbb{F}}}

\newcommand{\calA}{{\mathcal{A}}}
\newcommand{\calB}{{\mathcal{B}}}
\newcommand{\calC}{{\mathcal{C}}}
\newcommand{\calD}{{\mathcal{D}}}
\newcommand{\calE}{{\mathcal{E}}}
\newcommand{\calI}{{\mathcal{I}}}
\newcommand{\calJ}{{\mathcal{J}}}
\newcommand{\calL}{{\mathcal{L}}}
\newcommand{\calM}{{\mathcal{M}}}

\newcommand{\calX}{{\mathcal{X}}}
\newcommand{\calY}{{\mathcal{Y}}}
\newcommand{\calZ}{{\mathcal{Z}}}

\newcommand{\one}{{\mathbf{1}}}
\newcommand{\rd}{{\mathrm{d}}}

\newcommand{\half}{{\frac{1}{2}}}

\newcommand{\polyring}{{\mathfrak{R}}}

\newcommand{\dd}{{\mathsf{d}}} 

\newcommand{\ii}{{\mathbf{i}}}
\newcommand{\id}{{\mathbf{id}}}

\DeclareMathOperator{\tr}{{{tr}}}
\DeclareMathOperator{\Tr}{{Tr}}

\DeclareMathOperator*{\diam}{{diam}}
\DeclareMathOperator*{\Supp}{{Supp}}

\DeclareMathOperator*{\Mat}{{\mathsf{Mat}}}
\DeclareMathOperator*{\barMat}{{\overline{\Mat}}}

\DeclareMathOperator*{\mul}{{mul}}
\DeclareMathOperator{\dist}{{dist}}
\DeclareMathOperator{\Comm}{{Comm}}
\DeclareMathOperator{\Cent}{{Cent}}
\DeclarePairedDelimiter{\norm}{\lVert}{\rVert}
\DeclarePairedDelimiter{\abs}{|}{|}

\begin{document}
\title{Invertible subalgebras}
\author{Jeongwan Haah}
\address{Microsoft Quantum, Redmond, Washington, USA}
\address{Microsoft Quantum, Station Q, Santa Barbara, California, USA}
\begin{abstract}
We introduce invertible subalgebras of local operator algebras on lattices.
An invertible subalgebra is defined to be one 
such that
every local operator can be locally expressed 
by elements of the inveritible subalgebra and those of the commutant.
On a two-dimensional lattice,
an invertible subalgebra hosts a chiral anyon theory 
by a commuting Hamiltonian,
which is believed to be impossible on any full local operator algebra.
We prove that
the stable equivalence classes of $\mathsf d$-dimensional invertible subalgebras form an abelian group 
under tensor product,
isomorphic to the group of all $\mathsf d+1$ dimensional quantum cellular automata (QCA) 
modulo blending equivalence and shifts.

In an appendix,
we consider a metric on the group 
of all QCA on infinite lattices
and prove that the metric completion contains 
the time evolution by local Hamiltonians, 
which is only approximately locality-preserving.
Our metric topology is strictly finer than the strong topology.
\end{abstract}

\maketitle

Boundary algebras~\cite{GNVW,FreedmanHastings2019QCA,clifqca1}
on~$\ZZ^\dd$ of a quantum cellular automaton (QCA) on~$\ZZ^{\dd+1}$,
where a QCA is a $*$-automorphism of a local operator algebra that preserves locality,
have played an important role to understand QCA.
This in turn gives a nontrivial invariant 
of dynamics of periodically driven systems 
in~\mbox{$\dd+2$} dimensions; see~\cite{Harper2019} and references therein.
In this paper, we characterize boundary algebras of strictly locality-preserving QCA
as certain subalgebras, called invertible subalgebras, of local operator algebras.
The defining property is that each local operator~$x$
must be rewritten as~$x = \sum_i a_i b_i$, a sum of products of elements~$a_i$ 
of the invertible subalgebra and those~$b_i$ of its commutant, 
all near~$x$ in support.
See~\ref{def:invertible}.

The decomposition of the parent local operator algebra in $\dd$ dimensions
using an invertible subalgebra
allows us to define a QCA in one dimension higher.
It is our main result that such a QCA 
is unique up to blending equivalence~\cite{GNVW,FreedmanHastings2019QCA}.
Our construction is similar to that in one dimension~\cite{GNVW}
in that the constructed QCA is virtually a shift QCA
along one direction out of~$\dd+1$ directions of~$\ZZ^{\dd+1}$
and all the remaining $\dd$ directions are handled by
the inveritible subalgebra.
For this result, we may think of every QCA as a shift QCA in disguise,
pumping an invertible subalgebra along a coordinate axis.

The definition of invertible subalgebra 
is completely within $\dd$ dimensions 
in which the local operator algebra of interest lives in.
Hence, one may study them independently of QCA.
Local operator algebras and their norm-completions
have been studied with long history 
as a canonical mathematical model of quantum spin systems on lattices~\cite{BratteliRobinson2}.
They are studied abstractly also in the context of
approximately finite-dimensional $C^*$-algebras
or uniformly hyperfinite algebras,
usually ignoring the metric (locality) structure of the underlying lattice,
and there are classification results for subalgebras~\cite{Schafhauser2018}.
However, 
our notion of invertible subalgebras, where locality structure is important,
appears to be not considered before;
we remark that on finite systems
our notion of invertible subalgebra 
coincides with that of visibly simple algebra
of Freedman and Hastings~\cite{FreedmanHastings2019QCA}.

Invertible subalgebras and QCA can be studied using a sequence of finite systems,
by which one may perhaps avoid some technical difficulties 
arising from infinite dimensional algebras.
Such consideration may be necessary if one wishes to study these 
on a compact ``control space''
with nontrivial topology where lattice points live~\cite{FHH2019}.
In that setting, a sequence needs to be chosen so that 
objects are coherent in some sense.
The issue becomes somewhat subtle and especially important
if we want to think of the group structure of these objects,
where the number of group operations needs to be only finite 
with no obvious upper bound.
There is a resolution to these issues~\cite{FHH2019},
but in this paper we just consider~$\ZZ^\dd$, an infinite lattice, from the beginning.

The rest of the paper is organized as follows.
We formally define invertible subalgebras in~\S\ref{sec:setup}.
Next in~\S\ref{sec:Brauer}, 
we consider a stable equivalence relation among invertible subalgebras
and show that they form an abelian group.
With the locality structure aside, 
the construction mimics that of Brauer group of central simple algebras over a field.
Hence, we call the group of all stable equivalence classes of invertible subalgebras
a Brauer group.
By adapting an argument of~\cite{FreedmanHastings2019QCA},
we show that the Brauer group is trivial in one dimension.
In~\S\ref{sec:BoundaryAlgebra} we establish the connection between invertible subalgebras
and QCA, and prove our main result that 
they are the same objects up to certain equivalence relations.
In~\S\ref{sec:Pauli} we characterize invertible subalgebras
generated by a translation invariant set of Pauli operators.
This section is virtually a repackaging of~\cite[\S 3]{clifqca1}.
In fact, our definition of invertible subalgebras and 
the construction of QCA 
starting from an invertible subalgebra 
are generalizations of~\cite{clifqca1}.
In~\S\ref{sec:exampleTopologicalOrder}
we discuss a concrete example invertible subalgebra in two dimensions,
and show that in this subalgebra one can construct a commuting Hamiltonian
that realizes a chiral anyon theory.
This gives a piece of evidence that the example invertible subalgebra 
represents a nonzero Brauer class.
In~\S\ref{sec:discussion}, we summarize the results of this paper,
and briefly mention a blending equivalence of invertible subalgebras.
Finally, in~\S\ref{app:topo} 
we consider a metric topology on the space of all $*$-linear
transformations on a local operator algebra,
which is strictly finer than the strong topology~\cite{BratteliRobinson2}.
We derive a standard statement that time evolution by a local Hamiltonian
is differentiable in time, shown in our topology.
We prove that the time evolution 
is a limit of finite depth quantum circuits.
Perhaps, a nice scope of approximately locality-preserving $*$-automorphisms
may be obtained by taking closures of strictly locality-preseriving ones.
A closure depends on a chosen topology,
and the results in~\S\ref{app:topo}
will suggest that our metric is reasonable.

{\it Acknowledgments}:
I thank Lukasz Fidkowski, Matt Hastings, and Zhenghan Wang for encouraging discussions.
I also thank Corey Jones and David Penneys for lessons in AF algebras.
Some part of this work was done 
while I was attending a workshop 
at American Institute of Mathematics, San Jose,
on higher categories and topological order.
I thank the participants for discussions.

\tableofcontents

\section{Setup}\label{sec:setup}

Let $\Mat(\{\cdot\}, p)$ with an integer~$p \ge 1$ 
be the algebra of all $p \times p$ matrices 
over the field of complex numbers,
where an involution is given by the conjugate transpose~$\dag$.
As an abstract algebra, this is \emph{simple} (no two-sided ideal other than zero and the whole)
and \emph{central} (any element that commutes with every other element is a scalar multiple of the identity~$\one$).

Consider a $\dd$-dimensional lattice~$\ZZ^\dd$ where $\dd \ge 0$ is an integer.
By abuse of notation, let $p : \ZZ^\dd \to \ZZ_{>0}$ be a function,
called a {\bf local dimension assignment}.
We assign a matrix algebra~$\Mat(\{s\},p(s))$ on each lattice site~$s \in \ZZ^\dd$.
For any finite subset~$S \subset \ZZ^\dd$,
we consider a matrix algebra, denoted by~$\Mat(S,p) = \bigotimes_{s \in S} \Mat(s,p(s))$.
The collection of all those matrix algebras for all finite subsets of~$\ZZ^\dd$
is a directed system with embeddings~$\Mat(S,p) \hookrightarrow \Mat(S',p)$ 
for~$S \subseteq S'$
given by tensoring with identity on~$S' \setminus S$.
It is understood that~$\Mat(\emptyset,p) = \CC$.
The direct limit, which may be written as 
\begin{equation}
\Mat(\ZZ^\dd,p) = \bigcup_S \Mat(S,p),
\end{equation}
of this directed system is the {\bf local operator algebra} 
on $\ZZ^\dd$ with local dimension assignment~$p$.
The local dimension~$p(s)$ is finite for each site~$s$,
but may not be uniformly bounded across the lattice.
When this local dimension is not important (other than being finite everywhere)
we will often suppress~$p$ from notation and write~$\Mat(\ZZ^\dd)$.

For any operator $x \in \Mat(\ZZ^\dd)$,
there is some finite $S \subset \ZZ^\dd$ such that $x \in \Mat(S)$,
and the smallest such~$S$ is called the {\bf support} of~$x$,
and denoted by~$\Supp(x)$.
By definition, scalars ($c\one$ for some $c \in \CC$) have empty support.
The support of a subalgebra is the union of the support of all elements.
\emph{No} operator in $\Mat(\ZZ^\dd)$ has \emph{infinite} support,
but a subalgebra may have infinite support.
We will say that an operator or a subalgebra 
is {\bf supported on} some subset of the lattice
if their support is merely contained in it.
If $S \subseteq \ZZ^\dd$ is any subset and $\ell > 0$,
we denote by $S^{+\ell}$ the set of all sites
whose $\ell_\infty$-distance from~$S$ is at most~$\ell$.
A subalgebra of~$\Mat(\ZZ^\dd)$ is {\bf locally generated}
if, for some $\ell > 0$,
there exists a generating set for the subalgebra 
such that every generator has support of diameter at most~$\ell$.

\begin{definition}\label{def:invertible}
A unital $*$-subalgebra~$\calA$ of~$\Mat(\ZZ^\dd)$ is {\bf invertible}
if there exists a constant~$\ell > 0$, called spread, 
such that every $x \in \Mat(\ZZ^\dd)$
can be decomposed as
\begin{align}
x = \sum_i a_i b_i \text{ with some } 
\begin{cases}
a_i \in \calA, &\Supp(a_i) \subseteq \Supp(x)^{+\ell}\\
b_i \in \calB, &\Supp(b_i) \subseteq \Supp(x)^{+\ell}
\end{cases}
\end{align}
where $\calB = \calA'$ is the commutant of~$\calA$ within~$\Mat(\ZZ^\dd)$
and the sum over~$i$ is finite.
\end{definition}

Though we have set the stage with infinite dimensional algebras,
all phenomena we will see happen in finite systems
as long as there is a clear separation between the overall system size
and the spread length scale~$\ell$ that appears 
in the definition of QCA and invertible subalgebras.
Note also that we have not (yet) taken any metric completion of~$\Mat(\ZZ^\dd)$;
usually the operator norm is taken to obtain a {\bf quasi-local operator algebra} 
after completion.
We will consider some aspects of this completion in~\S\ref{app:topo}.
The tensor product of two $*$-algebras will always be the algebraic tensor product over~$\CC$.

\section{Brauer group of Invertible subalgebras} \label{sec:Brauer}

\begin{definition}[\cite{FreedmanHastings2019QCA}]
A $*$-subalgebra~$\calA$ of~$\Mat(\ZZ^\dd)$ is {\bf VS}
if there exists~$\ell > 0$, called a range, such that 
for any~$a \in \calA \setminus \CC\one$ and any site~$s \in \Supp(a)$
there exists $w \in \calA \cap \Mat(\{s\}^{+\ell})$ such that $[a,w] \neq 0$.
\end{definition}

This is a strong form of centrality
since any nontrivial element~$a$ of~$\calA$
has a local witness that $a$ is not in the center of~$\calA$,
near every point of~$\Supp(a)$.
Ref.~\cite{FreedmanHastings2019QCA} has introduced this notion,
and the authors called any such subalgebra ``visibly simple,'' hence VS.
This property by definition is not about simplicity but about centrality,
although for finite dimensional $*$-algebras over~$\CC$,
centrality coincides simplicity.
It is to be seen if this strong notion of centrality
implies simplicity in more general settings.
We will see in~\ref{prop:VS} 
that a VS subalgebra with an extra condition is invertible
and hence simple.

\begin{lemma}\label{lem:InvertibleCentral}
Let $\calA \subseteq \Mat(\ZZ^\dd)$ be any invertible subalgebra with spread~$\ell$,
and $\calB$ be its commutant within~$\Mat(\ZZ^\dd)$.
Then, both~$\calA$ and~$\calB$ are VS of range~$\ell$.
\end{lemma}

\begin{proof}
Let $a \in \calA \setminus \CC \one$ be arbitrary,
and let $s \in \Supp(a)$.
Since $\Mat(\{s\})$ is central,
there is~$x \in \Mat(\{s\})$ such that~$[a,x] \neq 0$.
Since $\calA$ is invertible, we have a decomposition $x = \sum_i a_i b_i$
where $b_i$ commutes with~$\calA$ and $\Supp(a_i) \subseteq \{s\}^{+\ell}$ for all~$i$.
Hence, some $a_i$ must not commute with~$a$.
This $a_i$ is what we wanted.

If $b \in \calB \setminus \CC \one$ and $t \in \Supp(b)$, 
then there is~$y \in \Mat(\{t\})$ such that~$[b,y] \neq 0$.
Since $\calA$ is invertible, we have a decomposition $y = \sum_j a'_j b'_j$
where $a'_j$ commutes with~$\calB$ and $\Supp(b'_j) \subseteq \{t\}^{+\ell}$ for all~$j$.
Hence, some $b'_j$ must not commute with~$b$.
This $b'_j$ is what we wanted.
\end{proof}

\begin{lemma}\label{lem:CommutantOfRestricted}
Let $\calA \subseteq \Mat(\ZZ^\dd)$ be any VS $*$-subalgebra with range~$\ell$.
Let $S \subseteq \ZZ^\dd$ be any set of sites.
Any element~$z \in \calA$ that commutes 
with every element of~$\calA \cap \Mat(S)$ is supported 
on~$(\ZZ^\dd \setminus S)^{+2\ell}$.
\end{lemma}

That is, any central element of the subalgebra~$\calA \cap \Mat(S)$
can only be supported at the boundary.

\begin{proof}
Suppose $\Supp(z)$ overlaps with the ``interior''~$S \setminus (\ZZ^\dd \setminus S)^{+2\ell}$
at, say, $s \in S$.
We have by the VS property
a ``witness''~$w \in \calA$ supported on an $\ell$-neighborhood of~$s$ 
such that~$[z,w] \neq 0$.
But the $\ell$-neighborhood of~$s$ is still in~$S$,
so $w \in \calA \cap \Mat(S)$.
But this is contradictory to the assumption that $z$ commutes with
every element of~$\calA \cap \Mat(S)$.
\end{proof}

\begin{lemma}\label{lem:VSInjective}
Let $\calA$ and $\calB$ 
be any unital mutually commuting $*$-subalgebras of~$\Mat(\ZZ^\dd)$.
If $\calA$ is VS, then the multiplication map
\begin{equation}
\mul : \calA \otimes \calB \ni \sum_i a_i \otimes b_i \mapsto \sum_i a_i b_i \in \calA \calB
\end{equation} 
is a $*$-algebra isomorphism.
\end{lemma}

\begin{proof}
The multiplication map is a well-defined $*$-algebra homomorphism 
since $\calA,\calB$ are mutually commuting.
By definition, this map~$\mul$ is onto~$\calA \calB$.
We have to show that $\mul$ is injective.
Suppose $c = \sum_i a_i \otimes b_i \in \ker \mul$ where the sum is finite.
Since both $\calA$ and $\calB$ consist of finitely supported operators,
each of~$a_i$ and~$b_i$ is finitely supported,
so there is some finite $S \subset \ZZ^\dd$ that support all~$a_i$ and~$b_i$.
We take $S = S(L_0)$ to be a box for some $L_0$,
where $S(L)$ is defined to be $[-L,L]^\dd \cap \ZZ^\dd$ for any $L \ge 0$.

Let $\calA_{2L} = \calA \cap \Mat(S(2L))$ and $\calB_{L} = \calB \cap \Mat(S(L))$.
Being finite dimensional semisimple, these decompose into simple $*$-subalgebras
by their centers~$Z(\calA_{2L}) = \mathrm{span}_\CC \{\pi_\mu\}$ 
and~$Z(\calB_{L}) = \mathrm{span}_\CC\{ \tau_\nu\}$ whose bases consisting of 
the minimal central projectors:
$\calA_{2L} \cong \bigoplus_\mu (\pi_\mu \calA_{2L})$ 
and $\calB_{L}  \cong \bigoplus_\nu (\tau_\nu \calB_{L})$.
We have $\calA_{2L} \otimes \calB_{L} \cong \bigoplus_{\mu, \nu} \pi_\mu \calA_{2L} \otimes \tau_\nu \calB_{L}$.
Each map~$\mul_{\mu,\nu}$, the restriction of $\mul$ on~$\pi_\mu \calA_{2L} \otimes \tau_\nu \calB_{L}$,
is injective whenever it is nonzero since the domain is simple.
Since the product~$\pi_\mu\tau_\nu$ and~$\pi_{\mu'}\tau_{\nu'}$ 
are orthogonal for any different tuples~$(\mu,\nu) \neq (\mu',\nu')$,
we see that $\mul$~restricted to $\calA_{2L} \otimes \calB_{L}$ is injective
if an only if $\mul_{\mu,\nu}$ are all injective.
Observe that 
$\mul_{\mu,\nu}(\pi_\mu \one \otimes \tau_\nu \one) = \pi_\mu \tau_\nu$,
where, by~\ref{lem:CommutantOfRestricted},
$\pi_\mu \in Z(\calA_{2L})$ is supported on the boundary of~$S(2L)$
and $\tau_\nu \in Z(\calB_{L})$ is supported somewhere in~$S(L)$.
For $L$ much larger than the range of the VS algebra~$\calA$,
the boundary of~$S(2L)$ is far from~$S(L)$,
and hence the product $\pi_\mu \tau_\nu$ cannot be zero.
Therefore, $\mul_{\mu,\nu}$ is nonzero and hence injective for those sufficiently large~$L$.
We conclude that $\mul : \calA_{2L} \otimes \calB_{L} \to \Mat(S(2L))$ is injective 
for all sufficiently large~$L$.

The kernel element~$c$ belongs to the kernel of~$\mul |_{\calA_{2L} \otimes \calB_{L}}$ 
for sufficiently large~$L$,
and hence, can only be zero.
\end{proof}

\begin{corollary}\label{lem:TensorIso}
Let $\calA$ be a unital $*$-subalgebra of~$\Mat(\ZZ^\dd)$ with the commutant~$\calB$.
Then,~$\calA$ is invertible if and only if
the multiplication map 
\begin{align}
	\mul : \calA \otimes \calB \to \Mat(\ZZ^\dd)
\end{align}
has a locality-preserving inverse in the sense of~\ref{def:invertible}.
\end{corollary}

\begin{proof}
If~$\calA$ is invertible, then it is VS by~\ref{lem:InvertibleCentral}.
Then, by~\ref{lem:VSInjective} the multiplication map is injective.
By definition of invertibility, the multiplication map is surjective.
The inverse obtained this way is locality preserving.
Conversely, a locality preserving inverse 
fulfills the definition of invertibility.
\end{proof}

\begin{proposition}\label{lem:csa}
Any invertible subalgebra~$\calA$ of~$\Mat(\ZZ^\dd)$ of spread~$\ell$
is central simple and $\ell$-locally generated.
\end{proposition}

\begin{proof}
First, let us show that $\Mat(\ZZ^\dd)$ is central simple and locally generated.
It is by definition locally generated.
If $m \in \Mat(\ZZ^\dd)$ is central,
then it is central in $\Mat(\Supp(m))$,
so $m \in \CC\one$.
For any nonzero element $x$ of any nonzero two-sided ideal~$\calI$ of~$\Mat(\ZZ^\dd)$,
the two-sided ideal~$X$ of~$\Mat(S)$ generated by~$x$ where $\Supp(x) \subseteq S$,
is nonzero, so $X = \Mat(S)$, and $\Mat(S) \subseteq \calI$.
Since $S$ is arbitrary, we must have $\calI = \Mat(\ZZ^\dd)$.

We have shown that $\calA$ is central in~\ref{lem:InvertibleCentral}.
Let~$\calB$ be the commutant of~$\calA$ within~$\Mat(\ZZ^\dd)$.
Let~$\calJ$ be a nonzero two-sided ideal of~$\calA$.
Then, $\calJ \calB \supseteq \calJ \one = \calJ$ is a nonzero two-sided ideal of a simple algebra~$\Mat(\ZZ^\dd)$,
so~$\calJ \calB = \Mat(\ZZ^\dd) \supseteq \calA$.
If~$a \in \calA$ is any element, 
$\mul(a \otimes \one) = a = \mul(\sum_j a_j \otimes b_j)$ 
where $a_j \in \calJ$, $b_j \in \calB$, and $\mul$ is the multiplication map.
Since~$\mul$ is injective by~\ref{lem:TensorIso},
this rearranges to $a \otimes \one - \sum_j a_j \otimes b_j = 0$.
We can rewrite this such that $b_j$'s are linearly independent,
implying that~$a$ is in the $\CC$-linear span of~$a_j \in \calJ$.
Therefore, $\calJ = \calA$ and $\calA$ is simple.

Let~$\calA_0$ be the collection of all $a_j$'s
that appear in the decomposition $x = \sum_j a_j b_j$, 
where $a_j \in \calA$ and $b_j \in \calB$,
of any single-site operator~$x \in \Mat(\ZZ^\dd)$.
By definition of invertible subalgebras,
$\calA_0$ consists of elements whose support size is uniformly bounded.
We claim that $\calA_0$ generates~$\calA$.
Let~$a \in \calA$ be arbitrary.
Since $a \in \Mat(\ZZ^\dd)$,
we have an $\calA \calB$-decomposition $a = \sum_i a'_i b'_i$.
Pulling it back to $\calA \otimes \calB$ by the multiplication map,
we see, as before, that~$a \in \mathrm{span}_\CC\{ a'_i \}$.
But every element of~$\Mat(\ZZ^\dd)$, 
not necessarily supported on a single site,
is a finite sum of finite products of single-site operators,
each of which $\calA \calB$-decomposes with elements of~$\calA_0$.
So, $a'_i \in \langle \calA_0 \rangle$, and therefore~$a \in \langle \calA_0 \rangle$.
\end{proof}

\begin{proposition}\label{lem:InvertibilityIsSymmetric}
If $\calA$ is any invertible subalgebra of~$\Mat(\ZZ^\dd)$ with the commutant~$\calB$,
then the commutant of~$\calB$, the double commutant of~$\calA$, is~$\calA$.
Therefore, a $*$-subalgebra of~$\Mat(\ZZ^\dd)$ is invertible if and only if 
its commutant is invertible.
\end{proposition}

\begin{proof}
Clearly, $\calA \subseteq \calB'$.
To show that $\calB' \subseteq \calA$, let~$x \in \calB'$ be arbitrary.
The invertibility of~$\calA$ implies that
$x = \mul( \sum_j a_j \otimes b_j)$ 
where $\mul: \calA \otimes \calB \to \Mat(\ZZ^\dd)$ is the multiplication map of~\ref{lem:TensorIso}.
We may assume that $a_j$ are linearly independent.
Since $x \in \calB'$ and $\mul$ is an isomorphism,
$\sum_j a_j \otimes b_j$ commutes with~$\one \otimes \calB$.
So, we have $\sum_j a_j \otimes [b,b_j] = 0$ for all~$b\in \calB$,
and the linear independence of~$a_j$ implies that $[b,b_j] = 0$ for each~$j$.
By~\ref{lem:InvertibleCentral} 
we know that~$\calB$ is central, so $b_j \in \CC \one$ for all~$j$.
Hence,~$x \in \mul( \sum_j a_j \otimes \CC\one) \subseteq \calA$.
\end{proof}

Recall that on~$\Mat$ we have a trace, a $\CC$-linear functional~$\tr$
such that $\tr(ab)=\tr(ba)$ and $\tr(a^\dag a) > 0$ for $a \neq 0$,
defined by~$\tr(x) = \Tr(x) / \dim_\CC S$
where $\Tr$ is the usual trace of a matrix
and $\dim_\CC S = \Tr(\one_S)$ means 
the dimension of the local algebra on~$S$ that supports~$x$.
It is important that this trace does not depend on the choice of~$S$.
In particular, $\tr(\one) = 1$.
The inequality~$\Tr(xy) \le \norm x \Tr(|y|)$ implies that
for any $x,y \in \Mat$,
\begin{align}
\tr(xy) &\le \norm x \tr(\abs y) \le \norm x \norm y.
\end{align}
Recall that for any positive semidefinite finite dimensional matrix~$\rho$
(a sum of operators of form~$z^\dag z$)
we have $\norm \rho = \sup_\sigma \Tr(\rho \sigma)$ 
where $\sigma$ ranges over all positive semidefinite matrices
with~$\norm{\sigma}_1 = \Tr \sigma \le 1$.
Writing in terms of normalized trace, $\tr$, on~$\Mat$,
we have
\begin{align}
	\norm \rho = \sup_{\sigma: \sigma \succeq 0, \tr(\sigma) \le 1} \tr(\rho \sigma),
\end{align}
the expression of which remains invariant under the embedding $\rho \mapsto \rho \otimes \one$,
and hence we may let~$\sigma$ range over all positive semidefinite operators of~$\Mat$
with $\tr \sigma \le 1$.

\begin{proposition}\label{prop:TrNormISA}
Let $\calA$ be an invertible subalgebra of~$\Mat(\ZZ^\dd)$
with the commutant~$\calB$.
Define
\begin{align}
\tr_\calA &: \Mat(\ZZ^\dd) \xrightarrow{\mul^{-1}\text{ by~\ref{lem:TensorIso}}} \calA \otimes \calB \xrightarrow{ \tr \otimes \id } \calB, \\
	\tr_\calB &: \Mat(\ZZ^\dd) \xrightarrow{\mul^{-1}\text{ by~\ref{lem:TensorIso}}} \calA \otimes \calB \xrightarrow{ \id \otimes \tr } \calA.\nonumber
\end{align}
Then, for all $x \in \calA$ and $y \in \calB$,
\begin{align}
	\tr(x y) &= \tr(\tr_\calA(xy)) = \tr(\tr_\calB(xy)) = \tr(x)\tr(y), \\
	\norm{xy} &= \norm x \norm y.\nonumber
\end{align}
\end{proposition}

\begin{proof}
For $L > 0$ we define $\calA_{2L} = \calA \cap \Mat( [-2L,2L]^\dd \cap \ZZ^\dd )$
and $\calB_{L} = \calB \cap \Mat( [-L,L]^\dd \cap \ZZ^\dd )$.
By~\ref{lem:InvertibleCentral} and~\ref{lem:CommutantOfRestricted},
we know the minimal central projectors $\{ \pi_\mu \} \subset \calA_{2L}$ and
$\{ \tau_\nu \} \subset \calB_L$ are supported near the boundary of the boxes~$[-(2)L,(2)L]^\dd$.
Hence, $\pi_\mu$ and $\tau_\nu$ have disjoint support for all sufficiently large~$L$,
implying $\tr(\pi_\mu \tau_\nu) = \tr(\pi_\mu) \tr(\tau_\nu)$.
For those large~$L$, so large that $\Supp(x)\cup \Supp(y) \subset [-L,L]^\dd$,
we have then a $*$-algebra isomorphism
$\pi_\mu \calA_{2L} \otimes \pi_\nu \calB_L \to \pi_\mu \calA_{2L} \tau_\nu \calB_L$
for each pair $(\mu,\nu)$.
For any fixed~$(\mu,\nu)$, a linear map
$\pi_\mu x  \otimes \tau_\nu y \mapsto \tr(\pi_\mu x)\tr(\tau_\nu y) / \tr(\pi_\mu \tau_\nu)$
is a normalized trace on $\pi_\mu \calA_{2L} \otimes \pi_\nu \calB_L$.
Also, a linear map~$\pi_\mu \pi_\nu z \mapsto \tr(\pi_\mu \pi_\nu z)/\tr(\pi_\mu \tau_\nu)$
is a normalized trace on~$\pi_\mu \calA_{2L} \pi_\nu \calB_L$.
Since a normalized trace is unique on any finite dimensional simple $*$-algebra,
they must agree: $\tr(\pi_\mu x)\tr(\tau_\nu y) = \tr(\pi_\mu x \tau_\nu y)$.
Summing over~$\mu$ and $\nu$,
we obtain the claim that $\tr(xy) = \tr(x)\tr(y)$.
It follows that $\tr(xy) = \tr(\tr_\calB(xy)) = \tr(\tr_\calA(xy))$.

For norm, it suffices to prove the claim for positive semidefinite~$x,y$,
because we will have
$\norm{xy}^2 = \norm{y^\dag x^\dag x y} = \norm{x^\dag x y^\dag y} = \norm{x^\dag x} \norm{y^\dag y} = \norm{x}^2 \norm{y}^2$.
So, assume $x, y \succeq 0$, implying $xy = yx \succeq 0$.
Let $\rho_x \in \Mat(\ZZ^\dd)$ be such 
that~$\rho_x \succeq 0$, $\tr(\rho_x) = 1$, and~$\norm x = \tr(x \rho_x)$;
such~$\rho_x$ exists because in~$\norm x = \sup_\rho \tr(x \rho)$
we can restrict the support of~$\rho$ to that of~$x$,
in which case $\rho$ ranges over a compact set.
We may further assume that $\rho_x \in \calA$
since $\tr(x \rho_x) = \tr(\tr_\calB(x \rho_x)) = \tr(x \tr_\calB(\rho_x))$.
Likewise, we find~$\rho_y \in \calB$ such that~$\rho_y \succeq 0$, $\tr(\rho_y) = 1$,
and~$\norm y = \tr(y \rho_y)$.
Then, $\rho_x \rho_y \succeq 0$ and $\tr(\rho_x \rho_y) = \tr(\rho_x)\tr(\rho_y) = 1$.
Therefore, $\norm{xy} \ge \tr(x y \rho_x \rho_y) = \tr(x \rho_x) \tr(y \rho_y) = \norm x \norm y$.
The opposite inequality~$\norm{x y} \le \norm{x} \norm{y}$ is trivial.
\end{proof}

\begin{definition}
Two $*$-subalgebras~$\calA_1 \subseteq \Mat(\ZZ^{\dd}, p_1)$ 
and~$\calA_2 \subseteq \Mat(\ZZ^{\dd}, p_2)$
are {\bf stably equivalent} (written~$\calA_1 \simeq \calA_2$)
if there exist
two local dimension assignments~$q_1,q_2 : \ZZ^\dd \to \ZZ_{>0}$,
a $*$-algebra isomorphism~$\phi$, and $\ell > 0$ 
such that
\begin{align}
\phi : \calA_1 \otimes \Mat(\ZZ^{\dd},q_1) &\xrightarrow{\quad \cong \quad} \calA_2 \otimes \Mat(\ZZ^{\dd}, q_2 ),\\
\Supp(\phi(x \otimes y)) &\subseteq \Supp(x \otimes y)^{+\ell} \text{ for all }
x \in \calA_1, y \in \Mat(\ZZ^\dd, q_1) .\nonumber
\end{align}
The stable equivalence class of an invertible subalgebra~$\calA$
among all invertible subalgebras of~$\Mat(\ZZ^{\dd})$
is denoted by~$[\calA]$, called the {\bf Brauer class} of~$\calA$.
A {\bf Brauer trivial} or {\bf stably trivial} 
invertible subalgebra is one that is stably equivalent 
to~$\CC = \Mat(\ZZ^{\dd}, \ZZ^{\dd} \to \{1\})$.
The set of all Brauer classes of invertible subalgebras of~$\Mat(\ZZ^\dd)$
is the {\bf Brauer group} of $\dd$-dimensional invertible subalgebras.
\end{definition}

Note that an isomorphism~$\phi$ that gives a stable equivalence
need not be defined on the entire~$\Mat(\ZZ^\dd,p_1)$.

\begin{proposition}\label{lem:BrauerGpInvertibleSubalgebras}
The set of all Brauer classes of invertible subalgebras on~$\ZZ^\dd$
is an abelian group under the tensor product operation:
$[\calA_1 \subseteq \Mat(\ZZ^\dd,p_1)] + [\calA_2 \subseteq \Mat(\ZZ^\dd,p_2)]$ 
is defined to be $[\calA_1 \otimes \calA_2 \subseteq \Mat(\ZZ^\dd,p_1 p_2)]$.
The inverse of~$[\calA]$ is represented by the commutant of~$\calA$
within~$\Mat(\ZZ^\dd)$ in which~$\calA$ resides.
\end{proposition}

\begin{proof}
The tensor product of two invertible subalgebras is an invertible subalgebra 
of the tensor product of the two parent algebras.
If $\calC_1 \otimes \Mat(\ZZ^\dd,c_1) \cong \calA_1 \otimes \Mat(\ZZ^\dd, a_1)$ 
and $\calC_2 \otimes \Mat(\ZZ^\dd, c_2) \cong \calA_2 \otimes \Mat(\ZZ^\dd, a_2)$,
then $\calC_1 \otimes \calC_2 \otimes \Mat(\ZZ^\dd, c_1 c_2) \cong \calA_1 \otimes \calA_2 \otimes \Mat(\ZZ^\dd, a_1 a_2)$.
This means that the group operation is well defined.
The claim on the inverse is already proved in~\ref{lem:TensorIso}.
\end{proof}

Let us remark on the relation between invertible and VS subalgebras.

\begin{proposition}\label{prop:VS}
Let $\calA \subseteq \Mat(\ZZ^\dd)$ be any unital VS $*$-subalgebra
with the commutant~$\calB$ within $\Mat(\ZZ^\dd)$.
Suppose that the multiplication 
map~$\mul : \calA \otimes \calB \to \Mat(\ZZ^\dd)$ is surjective.
Then, $\calA$ is invertible.
\end{proposition}

We did not impose the locality condition of invertibility,
but instead assumed the VS property of~\cite{FreedmanHastings2019QCA}.
The assumption on~$\mul$ is automatic in any finite lattice,
so the VS property coincides with our invertibility.
We have shown in~\ref{lem:InvertibleCentral} and~\ref{lem:TensorIso} 
that the converse holds generally.
Ref.~\cite{FreedmanHastings2019QCA} posed an open question
if the commutant of any VS subalgebra is VS.
We answer it yes by~\ref{lem:InvertibilityIsSymmetric} for any finite systems.

\begin{proof}
We have to find a decomposition~$x = \sum_i a_i b_i$ for any~$x \in \Mat(\ZZ^\dd)$
such that~$a_i \in \calA \setminus \CC\one$ and~$b_i \in \calB$ are uniformly local.
A decomposition with no locality property exists since~$\mul$ is surjective.
Assume without loss of generality that $b_i$ are linearly independent.
For any~$w \in \calA$ with $\Supp(w) \cap \Supp(x) = \emptyset$,
we have $0 = [w,x] = \sum_i [w,a_i] b_i = \mul(\sum_i [w,a_i] \otimes b_i)$.
Since $\mul$ is injective by~\ref{lem:VSInjective}, we have $[w,a_i] = 0$ for all~$i$.
Since $\calA$ is VS with range~$\ell$,
we see $\Supp(a_i) \subseteq \Supp(x)^{+10\ell}$ for all~$i$.
We have shown that for any~$x$
there is a decomposition~$x \sum_i a_i b_i$ where $a_i \in \calA$ are all near~$x$ and 
are linearly independent,
and $b_i \in \calB$.

Now, rearrange the sum~$x = \sum_i a_i b_i$ if needed,
so that $a_i$ are linearly independent.
If $\Supp(b_{i_0})$ for some~$i_0$ contains a site~$s$ 
that is distance~$50\ell$ away from~$\Supp(x)$,
then there is an element~$y \in \Mat(\{s\})$ 
such that $[y,b_{i_0}] \neq 0$.
By the previous paragraph,
$y$ also has a decomposition~$y = \sum_j a'_j b'_j$ where~$a'_j \in \calA$
are supported within $10 \ell$-neighborhood of~$s$ and are linearly independent,
and~$b'_j \in \calB$.
Then,
$0 = [x,y] = \sum_{i,j} a_i a'_j [b_i, b'_j] = \mul(\sum_{i,j} a_i a'_j \otimes [b_i,b'_j])$.
By~\ref{lem:VSInjective} again,
we must have $\sum_{i,j} a_i a'_j \otimes[ b_i, b'_j  ] = 0$.
Since $a_i$ and $a'_j$ are far apart in support,
the products~$\{a_i a'_j\}_{i,j}$ are linearly independent, too,
and we must have $[b_i,b'_j] = 0$ for all $i,j$.
Then, $[b_{i_0}, y] = 0$, a contradiction.
So, $\Supp(b_i)$ for all~$i$ are contained within distance~$50\ell$ from~$x$.
\end{proof}

Freedman and Hastings~\cite{FreedmanHastings2019QCA}
have proved that one-dimensional VS $*$-subalgebras 
in any finite lattice are ``trivial,''
generated by a collection of mutually commuting local central simple $*$-subalgebras.
Indeed, 
the Brauer group on the infinite one-dimensional lattice~$\ZZ$ is trivial.
The proof of this result borrows much from the proof of~\cite[Thm.~3.6]{FreedmanHastings2019QCA}
for finite systems with periodic boundary conditions,
and is presented in~\S\ref{app:1d}.

\section{Boundary algebras of QCA}\label{sec:BoundaryAlgebra}

\begin{definition}\label{def:QCA}
A $*$-automorphism~$\alpha$ of~$\Mat(\ZZ^\dd, p)$ is a {\bf QCA} if
there exists a constant~$\ell > 0$, called spread,
such that for all $x \in \Mat(\ZZ^\dd, p)$ we have
\begin{align}
\Supp(\alpha(x)) \subseteq \Supp(x)^{+\ell}.
\end{align}
A QCA~$\alpha$ is a {\bf shift} 
if for any site~$s \in \ZZ^\dd$ there is a 
factorization~$\Mat(\{s\}, p(s)) \cong \bigotimes_j \Mat(\{s\}, p_{s,j})$
for some positive integers~$p_{s,j}$
such that for any~$j$ the image $\alpha(\Mat(\{s\}, p_{s,j}))$ is supported on a single site.
A shift QCA~$\alpha$ is {\bf monolayer}
if each single-site algebra~$\Mat(\{s\}, p(s))$ 
is mapped onto a single-site algebra.
A shift QCA is bilayer or trilayer, {\it etc.}, 
if it is a tensor product of two or three, {\it etc.} monolayer shift QCA.
\end{definition}

\begin{lemma}[\cite{Arrighi2007}]\label{lem:QCAInverseIsQCA}
If~$\alpha$ is a QCA with spread~$\ell$, then~$\alpha^{-1}$ is a QCA with spread~$\ell$.
\end{lemma}
\begin{proof}
For any element $x \in \Mat(\ZZ^\dd)$ and any subset~$S \subseteq \ZZ^\dd$,
we have that $\Supp(x) \subseteq S$ if and 
only if~$[x,y] = 0$ for all~$y \in \Mat(\ZZ^\dd \setminus S)$.
So, the condition for QCA is equivalent to~$[\alpha(x), y] = 0$
for any~$x$ and~$y$ whose supports are~$\ell$-apart.
But, this is equivalent to~$[x,\alpha^{-1}(y)] = 0$ for all such~$x,y$.
\end{proof}

\begin{lemma}\label{lem:TensorFactor}
If a $*$-subalgebra~$\calM$ of~$\Mat(\ZZ^\dd)$ 
contains $\Mat(T)$ for some~$T \subseteq \ZZ^\dd$,
then the multiplication map
\begin{equation}
	\mul: \Mat(T) \otimes (\calM \cap \Mat(\ZZ^\dd \setminus T)) \to \calM
\end{equation}
is a $*$-algebra isomorphism.
\end{lemma}

\begin{proof}
The map is an injective $*$-homomorphism because the tensor factors in the domain
have disjoint support.
We have to show that $\mul$ is surjective.
Put $S = \ZZ^\dd \setminus T$.
Let $x \in \calM$ be arbitrary.
Using $\Mat(\ZZ^\dd) = \Mat(T) \otimes \Mat(S)$,
we have $x = \sum_i t_i \otimes s_i$ 
where~$t_i \in \Mat(T)$ and~$s_i \in \Mat(S)$.
Since all~$t_i$ are finitely supported, say on a common finite set~$R \subseteq T$,
we may expand $t_i = \sum_j \lambda_{ij} u_j$ where $u_j$ are unitary orthonormal elements
under the Hilbert--Schmidt inner product defined on~$\Mat(R)$ and~$\lambda_{ij} \in \CC$.
Then, $x = \sum_i (\sum_j \lambda_{ij}u_j) \otimes s_i = \sum_j u_j \otimes (\sum_i \lambda_{ij} s_i)$.

We claim that each~$s'_j = \sum_i \lambda_{ij} s_i \in \Mat(S)$ 
is actually in~$\calM \cap \Mat(S)$;
this will show the surjectivity.
Since $R$ is finite, 
we can consider the Haar probability measure~$\rd v$ on
the unitary group~$\{v\}$ of~$\Mat(R) \subseteq \calM$.
Every $v$ commutes with~$\Mat(S)$.
Conjugation-integral allows us to evaluate the Hilbert--Schmidt inner product:
$\int v (u_k^\dag u_j) v^\dag \rd v = \delta_{kj} \one$
where $\delta_{kj}$ is the Kronecker delta.
Then,
\begin{equation}
\calM \ni \int v u_k^\dag x v^\dag \rd v
= \int v u_k^\dag (\sum_j u_j s'_j) v^\dag \rd v
= \sum_j \int  v u_k^\dag u_j v^\dag s'_j \rd v
= \sum_j \delta_{jk}\one s'_j = s'_k . \qedhere
\end{equation}
\end{proof}

\begin{lemma}[\cite{GNVW,FreedmanHastings2019QCA}]\label{lem:BoundaryAlgebra}
Let $\alpha$ be a QCA with spread~$\ell > 1$ on~$\Mat(\ZZ^\dd)$.
For any~$n \in \ZZ$, we define
\begin{align}
	H(n) &= (\ZZ \cap (-\infty, n]) \times \ZZ^{\dd-1} & \text{(half lattice)},\nonumber\\
	T(n,\ell) &= (\ZZ \cap (-\infty, n-\ell]) \times \ZZ^{\dd-1} &\text{(``interior'' of }H_n),\label{eq:BoundaryAlgebra}\\
	S(n,\ell) &= (\ZZ \cap [n-\ell+1,n+\ell]) \times \ZZ^{\dd-1} &\text{(``boundary'' of }H_n),\nonumber\\
	\calB(n,\ell) &= \alpha(\Mat(H(n))) \cap \Mat(S(n,\ell)).\nonumber
\end{align}
Then, $\calB(n,\ell)$ is a unital $*$-subalgebra of~$\Mat(S(n,\ell))$ 
such that the multiplication map
\begin{equation}
	\mul : \Mat(T(n,\ell)) \otimes \calB(n,\ell) \to \alpha(\Mat(H(n)))
\end{equation}
is a $*$-algebra isomorphism.
\end{lemma}
\begin{proof}
We suppress $n$ and $\ell$ from notation in this proof.
Observe that $\alpha(\Mat(H))$ is contained in~$\Mat(T \sqcup S)$.
Since $\alpha^{-1}$ is a QCA with the same spread by~\ref{lem:QCAInverseIsQCA},
we have $\alpha^{-1}(\Mat(T)) \subseteq \Mat(H)$
and therefore $\Mat(T) \subseteq \alpha(\Mat(H))$.
Apply~\ref{lem:TensorFactor} to conclude the proof.
\end{proof}

\begin{definition}
The subalgebra~$\calB(n,\ell)$ in~\ref{lem:BoundaryAlgebra} is called a {\bf boundary algebra}
of~$\alpha$ on the positive first axis.
If the slab~$S(n,\ell)$ in~\ref{lem:BoundaryAlgebra} is taken orthogonal to the~$j$-th direction,
then we call it a boundary algebra on the positive~$j$-th axis.
\end{definition}

\begin{lemma}\label{lem:BoundaryAlgebraIsInvertible}
Every boundary algebra of a QCA on~$\ZZ^\dd$ is an invertible subalgebra on~$\ZZ^{\dd-1}$.
If the QCA has spread at most~$\ell$, 
then the invertible subalgebra has spread at most~$2\ell$.
\end{lemma}

\begin{proof}
We use the notations from the statement of~\ref{lem:BoundaryAlgebra}.
Let~$\calA$ be the commutant of~$\calB$ within~$\Mat(S)$.
Let~$y \in \Mat(S)$ be arbitrary.
Then, $\alpha^{-1}(y)$ is supported on a thickened slab 
that straddles between~$H$ and~$\ZZ^\dd \setminus H$.
Using~$\Mat(\ZZ^\dd) = \Mat(H) \otimes \Mat(\ZZ^\dd \setminus H)$,
we obtain~$\alpha^{-1}(y) = \sum_i b'_i \otimes a'_i$ 
where~$b'_i \in \Mat(H)$ and~$a'_i \in \Mat(\ZZ^\dd \setminus H)$.
By~\ref{lem:QCAInverseIsQCA}, $\alpha^{-1}(y)$ is supported 
on the $\ell$-neighborhood of~$\Supp(y)$,
and hence each $b'_i$ and $a'_i$ can be chosen to be within the same neighborhood.
Put~$b_i = \alpha(b'_i) \in \alpha(\Mat(H))$ and~$a_i = \alpha(a'_i)$;
they are supported within distance $2\ell$ from~$\Supp(y)$.
Then, $y = \sum_i b_i a_i$.
This is almost what we want for the invertibility
except that we have to show that~$b_i \in \calB$ and~$a_i \in \calA$.

Let $A$ be a finite set of sites outside~$S$ such that $S \sqcup A$ supports all~$a_i$.
We can choose $A$ not to intersect~$T$ 
since~$a'_i$ is on a region with the first coordinate larger than~$n$,
whose $\ell$-neighborhood cannot reach~$T$.
Similarly, there is a finite set~$B$ of sites contained in~$T$
such that $B \sqcup S$ supports all~$b_i$.
Let $\{u\}$ be the finite dimensional unitary group of~$\Mat(A)$
and $\{v\}$ be that of~$\Mat(B)$.
By construction, $[u,v] = 0$ for all~$u,v$,
and moreover $[u,b_i]  = 0$ and $[v,a_i] = 0$ for their supports are disjoint.
Then, the projection 
by a Haar integral~$y  = \iint u v y u^\dag v^\dag \rd u \rd v = \sum_i (\int v b_i v^\dag \rd v) (\int u a_i u^\dag \rd u)$ gives a desired decomposition for the invertibility
because the projection  does not enlarge the support within~$S$.
\end{proof}

The spread of a QCA is only an upper bound on 
how much the support of the image of
an element grows.
Hence, when taking a boundary algebra according to~\ref{lem:BoundaryAlgebra},
a larger $\ell$ may be used, in which case the resulting subalgebra in~$\dd-1$ dimensions
will be larger.
Specifically, from~\eqref{eq:BoundaryAlgebra} we see that
the larger boundary algebra is always of form~$\calB \otimes \Mat(R \times \ZZ^{\dd-1})$
for some finite set~$R$ of coordinates on the first axis.
The extra tensor factor is a local operator algebra on~$\ZZ^{\dd-1}$
where each site is now occupied with a larger qudit that combines all qudits along~$R$.
Hence, a larger~$\ell$ gives a stably equivalent boundary invertible subalgebra.
Furthermore, we have the following.

\begin{lemma}\label{lem:BoundaryAlgebraIndependentOfCut}
Given any $\dd$-dimensional QCA~$\alpha$,
the boundary invertible algebras 
corresponding to different~$n$ (the location of the boundary) in~\ref{lem:BoundaryAlgebra}
are stably equivalent to one another.
\end{lemma}

\begin{proof}
We use the same notation as in~\eqref{eq:BoundaryAlgebra}.
We know that $[\calB(n,\ell)] = [\calB(n,\ell')]$ for any~$n$ and any~$\ell' > \ell$.
For sufficiently large~$\ell'$,
we see that $\alpha^{-1}(\calB(n+1,\ell'+1))$ 
contains~$\Mat(P)$, where $P = H(n+1) \setminus H(n)$.
By~\ref{lem:TensorFactor},
$\alpha^{-1}(\calB(n+1,\ell'+1))$ has~$\Mat(P)$ as a tensor factor.
On the other hand, 
by inspection of the definition,
$\calB(n+1,\ell'+1)$ is nothing but~$\calB(n,\ell') \alpha(\Mat(P))$.
Therefore, $\alpha^{-1}(\calB(n+1, \ell'+1)) \cong \alpha^{-1}(\calB(n,\ell')) \otimes \Mat(P)$.
Hence,
\begin{align}
\calB(n,\ell) &\simeq \calB(n,\ell') \simeq \alpha^{-1}(\calB(n,\ell'))
\simeq \alpha^{-1}(\calB(n,\ell')) \otimes \Mat(P)\\
&\simeq \alpha^{-1}(\calB(n+1,\ell'+1))
\simeq \calB(n+1,\ell'+1)
\simeq \calB(n+1,\ell).\qedhere\nonumber
\end{align}
\end{proof}

Therefore, we may speak of~\emph{the} Brauer class of boundary algebras
of a QCA on a positive axis.
We have little idea 
on how a different choice of an axis on which we take a boundary algebra
affects the Brauer class of the boundary invertible subalgebra.
Indeed, we may not expect a complete spatial isotropy
since a boundary algebra on the \emph{negative} axis 
is the Brauer inverse of that on the positive axis,
as shown in~\ref{lem:InvertibleSubalgebraToQCA} below.
In~\S\ref{sec:discussion} we will remark that
boundary algebras on different axis ``blend'' in a certain sense.

\begin{lemma}\label{lem:BoundaryAlgebraOnNegative}
Using the notation from~\eqref{eq:BoundaryAlgebra},
we have
\begin{align}
\calA(n,\ell)\calB(n,\ell) = \Mat(S(n,\ell)) 
\quad \text{ where }\quad \calA(n,\ell) = \alpha( \Mat(\ZZ^\dd \setminus H(n)) ) \cap \Mat(S(n,\ell)).
\end{align}
\end{lemma}

\begin{proof}
By construction, we have $\calA(n,\ell)\calB(n,\ell) \subseteq \Mat(S(n,\ell))$.
Considering half spaces towards positive infinity,
the result of~\ref{lem:BoundaryAlgebra} is recast as
\begin{align}
\Mat(T(n,\ell)) \calB(n,\ell) &= \alpha(\Mat(H(n)))\\
\calA(n,\ell) \Mat(\ZZ^\dd \setminus (T(n,\ell)\cup S(n,\ell)) ) &= \alpha(\Mat(\ZZ^\dd \setminus H(n))).\nonumber
\end{align}
Multiplying these two, since $\alpha$ is surjective,
we see that $\calA(n,\ell)\calB(n,\ell) \supseteq \Mat(S(n,\ell))$.
\end{proof}

\begin{lemma}\label{lem:InvertibleSubalgebraToQCA}
Given any invertible subalgebra~$\calB$ in $\Mat(\ZZ^{\dd-1},p)$,
there exists a local operator algebra~$\Mat(\ZZ^\dd, 1 \times p)$
and a QCA~$\alpha$ on it, of which~$\calB$ is stably equivalent to a boundary algebra.
If $\calB$ has spread at most~$\ell$, then $\alpha$ has spread at most~$\ell$.
\end{lemma}

\begin{proof}
Construct a lattice in $\dd$ dimensions by copying that of $\Mat(\ZZ^{\dd-1})$
at each point of a new axis~$\ZZ$, 
and construct the local operator algebra~$\Mat(\ZZ^\dd)$.
Declare the new axis the first axis.
The local dimension assignment is translation-invariant along the first axis.
Let~$\calA$ be the commutant of~$\calB$ within~$\Mat(\ZZ^{\dd-1})$.
Let $j \in \ZZ$.
For any $x^{(j)} \in \Mat(\{j\} \times \ZZ^{\dd-1})$, 
we use the invertibility of~$\calB$ to write $x^{(j)} = \sum_i a_i^{(j)} b_i^{(j)}$
where~$a_i^{(j)} \in \calA$ and $b_i^{(j)} \in \calB$ are supported 
on the $\ell$-neighborhood of~$\Supp(x^{(j)}) \subseteq \{j\} \times \ZZ^{\dd-1}$.
Then, we define
\begin{align}
	\alpha_j :  x^{(j)} &\mapsto \sum_i a_i^{(j)} \otimes b_i^{(j+1)} \in \Mat(\{j\} \times \ZZ^{\dd-1}) \otimes \Mat(\{j+1\} \times \ZZ^{\dd-1}),\\
	\beta_j :  x^{(j)} &\mapsto \sum_i a_i^{(j)} \otimes b_i^{(j-1)} \in \Mat(\{j\} \times \ZZ^{\dd-1}) \otimes \Mat(\{j-1\} \times \ZZ^{\dd-1}).\nonumber
\end{align}
By~\ref{lem:TensorIso}, each~$\alpha_j$ and each~$\beta_j$ 
is a well defined $*$-algebra homomorphism.
Next, we define 
\begin{align}
	\alpha : \Mat(\ZZ^\dd) \ni x = \prod_j x^{(j)} &\mapsto \prod_j \alpha_j(x^{(j)}) \in \Mat(\ZZ^\dd),\\
	\beta : \Mat(\ZZ^\dd) \ni x = \prod_j x^{(j)} &\mapsto \prod_j \beta_j(x^{(j)}) \in \Mat(\ZZ^\dd)\nonumber
\end{align}
and extend to the entire~$\Mat(\ZZ^\dd)$ by $\CC$-linearity.
Here, the products are always finite and are well defined 
because~$\calA$ and~$\calB$ mutually commute.
Both~$\alpha$ and~$\beta$ preserve multiplication 
because each~$\alpha_j$ and each~$\beta_j$ does.
Moreover, $\beta \circ \alpha = \id = \alpha \circ \beta$:
\begin{align}
	x^{(j)} 
	&\xmapsto{\quad \alpha \quad } \sum_i a_i^{(j)} b_i^{(j+1)}
	\xmapsto{\quad \beta \quad } \sum_i \beta_j(a_i^{(j)})\beta_{j+1}(	b_i^{(j+1)}) 
	= \sum_i a_i^{(j)} b_i^{(j)} = x^{(j)},\\
	x^{(j)} 
	&\xmapsto{\quad\beta\quad } \sum_i a_i^{(j)} b_i^{(j-1)}
	\xmapsto{\quad\alpha\quad } \sum_i \alpha_j(a_i^{(j)})\alpha_{j-1}(	b_i^{(j-1)}) 
	= \sum_i a_i^{(j)} b_i^{(j)} = x^{(j)}.\nonumber
\end{align}
Therefore, $\alpha$ is a QCA with spread~$\ell$.

It remains to show that~$\calB$ 
is stably equivalent to a boundary algebra of~$\alpha$ on the positive first axis.
But this is obvious, 
since $\calM = \Mat((\ZZ \cap (-\infty,0])\times \ZZ^{\dd-1})$
is mapped under~$\alpha$ to $\calM \otimes \calB$.
\end{proof}

\begin{definition}
Let~$\alpha$ be a QCA of~$\Mat(\ZZ^\dd,p)$
and~$\beta$ be a QCA of~$\Mat(\ZZ^\dd,q)$.
For any $\ell > 0$, let $S_-(n,\ell) = (\ZZ \cap (-\infty,n - \ell)) \times \ZZ^{\dd-1}$
and $S_+(n,\ell) = (\ZZ \cap (n + \ell,\infty)) \times \ZZ^{\dd-1}$.
We say that $\alpha$ {\bf blends into}~$\beta$,
along the first axis
if 
there exists~$\ell > 0$ (unrelated to the spread of the two QCA)
and $n \in \ZZ$ (position of blending interface)
and a QCA~$\gamma$ of~$\Mat(\ZZ^\dd,r)$ (a blending QCA)
such that
\begin{align}
r(s) = p(s) \text{ and } \gamma(x) = \alpha(x) 
&\text{ for all }
x \in \Mat(S_-(n,\ell),p),
s \in S_-(n,\ell),\\
r(s) = q(s) \text{ and } \gamma(x) = \beta(x) 
&\text{ for all }
x \in \Mat(S_+(n,\ell),q),
s \in S_+(n,\ell).\nonumber
\end{align}
\end{definition}

\begin{definition}
A depth~1 quantum circuit with spread~$\ell$ 
is an $*$-automorphism of~$\Mat(\ZZ^\dd)$ defined by
a collection~$\{u_a \in \Mat(\ZZ^\dd) \}_a$ of local unitaries
such that $\diam\Supp(u_a) \le \ell$ and $\Supp(u_a) \cap \Supp(u_b) = \emptyset$ whenever $a \neq b$,
and the action is given by conjugation.
A finite depth quantum circuit, or {\bf quantum circuit} for short,
is a composition of finitely many depth~1 quantum circuits.
Every quantum circuit is a QCA.
\end{definition}

Every quantum circuit blends into the identity anywhere,
and the identity blends into any quantum circuit.
Just drop the local unitaries!

\begin{lemma}\label{lem:BlendingIsEquivalence}
Along the first axis,
a QCA~$\alpha$ blends into a QCA~$\beta$
with blending interface at~$n$,
if and only if~$\alpha$ blends into~$\beta$
with blending interface anywhere,
if and only if~$\beta$ blends into~$\alpha$ with blending interface anywhere.
\end{lemma}

\begin{proof}
It is well-known~\cite{Arrighi2007} 
that $\alpha \otimes \alpha^{-1}$ is a finite depth quantum circuit.
(\emph{Proof}:
$\alpha \otimes \alpha^{-1} = (\alpha \otimes \id) \mathrm{SWAP} (\alpha^{-1} \otimes \id) \mathrm{SWAP}$, but $\mathrm{SWAP}$ is a finite depth quantum circuit,
and $(\alpha \otimes \id) \mathrm{SWAP} (\alpha^{-1} \otimes \id)$
is a product of commuting gates each of which is an image of local SWAP operator under~$\alpha$,
which can be turned into finitely many layers of nonoverlapping gates on~$\ZZ^\dd$.)
Observe that a finite depth quantum circuit 
blends into the empty QCA~$\one : \CC \to \CC$ 
on~$\CC=\Mat(\ZZ^\dd, p=1)$ with blending interface placed anywhere.

To show the first ``iff,'' suppose~$\alpha$ blends into~$\beta$ at~$n$,
which we write as~$\alpha \xrightarrow{\sim n \sim} \beta$. 
We first move the interface to, say,~$n' \ll n$.
Observe that $\alpha \otimes \beta^{-1} \xrightarrow{\sim n \sim } \one$.
Take the slab~$(\ZZ \cap [n',n]) \times \ZZ^{\dd-1}$
and regard the local operator algebra on this slab 
as a $(\dd-1)$-dimensional local operator algebra.
Then we can say that $\alpha \otimes \beta^{-1} \xrightarrow{\sim n' \sim } \one$;
the blending QCA has an increased local dimension at the interface.
So, $\alpha \otimes \beta^{-1} \otimes \beta \xrightarrow{\sim n' \sim} \beta$.
Since $\one \xrightarrow{\sim n' \sim} \beta^{-1} \otimes \beta$,
we have $\alpha \xrightarrow{\sim n' \sim} \beta$.
We can similarly move the interface to $n' \gg n$
by starting with $\one \xrightarrow{\sim n \sim} \beta \otimes \alpha^{-1}$.

To show the second ``iff,'' we use the fact that a blending QCA is invertible:
$\alpha \xrightarrow{\sim n \sim} \beta$ 
implies
$\alpha^{-1} \xrightarrow{\sim n \sim} \beta^{-1}$,
which implies
$\alpha^{-1} \otimes \alpha \otimes \beta \xrightarrow{\sim n \sim} 
\beta^{-1} \otimes \alpha \otimes \beta$,
which implies
$\beta \xrightarrow{\sim n \sim} \alpha$.
\end{proof}

\begin{corollary}
Given an axis, the blending relation is an equivalence relation,
modulo which the set of all $\dd$-dimensional QCA form an abelian group under tensor product,
with the identity class~$[\id]$ represented by the identity QCA
and the inverse of~$[\alpha]$ represented by~$\alpha^{-1}$.
\end{corollary}

\begin{proof}
Clear by~\ref{lem:BlendingIsEquivalence} 
and the fact that $\alpha \otimes \alpha^{-1}$ blends into the identity.
\end{proof}

Applying~\ref{lem:BoundaryAlgebraIndependentOfCut} to a blending QCA,
we see that the Brauer class of boundary algebras on the first axis
is invariant under blending relation.
Conversely,

\begin{lemma}\label{lem:TrivialBoundaryAlgebraImpliesBlendedTrivialQCA}
A QCA with a Brauer trivial boundary algebra on the positive first axis,
blends into a bilayer shift QCA along the first axis.
\end{lemma}

\begin{proof}
We use the notations of~\eqref{eq:BoundaryAlgebra} in this proof.
Without loss of generality we assume that $\ell = 1$
so the image of any single-site operator under~$\alpha$
is supported on immediately neighboring sites,
with respect to the $\ell_\infty$-distance on~$\ZZ^\dd$.
This can be achieved simply by ``compressing the lattice'':
we consider a block of sites as a new site,
and form a lattice of the blocks.
Though we set $\ell =1$, we keep the notation~$\ell$ 
as long as its specific value is not important.

The Brauer triviality implies that 
a representative~$\calA(n,\ell)$ of the Brauer inverse of~$[\calB(n,\ell)]$
is also Brauer trivial by~\ref{lem:BrauerGpInvertibleSubalgebras}.
Here,~$\calA(n,\ell)$ is defined in~\ref{lem:BoundaryAlgebraOnNegative},
and is the commutant of~$\calB(n,\ell)$ within~$\Mat(S(n,\ell))$.
So, we have a locality-preserving $*$-algebra isomorphism
\begin{align}
\phi &: \Mat(\ZZ^{\dd-1}, q_\mathrm{left}) 
\xrightarrow{\quad \cong \quad} 
\calA(n,\ell) \otimes \Mat(\ZZ^{\dd-1},p')
\end{align}
for some local dimension assignments~$q_\mathrm{left}$ and~$p'$.

Now we define two monolayer shift QCA~$\beta_\mathrm{left}, \beta_\mathrm{right}$.
The QCA~$\beta_\mathrm{left}$ is defined to be
a translation to the left along the first axis by one lattice unit (distance~$1$)
acting on~$\Mat(\ZZ^\dd, 1 \times q_\mathrm{left})$.
Here, $1 \times q_\mathrm{left}$ means that this sheet is repeated along the first axis.
To define $\beta_\mathrm{right}$,
we note that $S(n,\ell) = S(n,1) = \{n,n+1\} \times \ZZ^{\dd-1}$.
So, we have a local dimension assignment~$p(n+1, \cdot)$ on a $(\dd-1)$-dimensional sheet.
The QCA~$\beta_\mathrm{right}$ is defined to be 
a translation to the right along the first axis by one lattice unit
acting on~$\Mat(\ZZ^\dd, 1 \times (p(n+1,\cdot) p') )$.
Here, $p'$ appears in the stabilization for~$\phi$.

Finally, we blend~$\alpha$ into~$\beta = \beta_\mathrm{left} \otimes \beta_\mathrm{right}$.
Let us use interval notations to denote various set of sites:
instead of~$(\ZZ \cap (a,b]) \times \ZZ^{\dd-1}$,
we simply write $(a,b]$.
A blending QCA is defined by Table~\ref{tb:blending}.
It is well defined because any commuting elements are mapped to commuting elements.
The injectivity is automatic since~$\Mat(\ZZ^\dd)$ is simple.
The surjectivity is checked by inspection of the table as follows.
Every sheet, a hyperplane of~$\ZZ^\dd$ orthogonal to the first axis,
in~$(-\infty,n-1]$ is reached by~$\alpha(x)$ for some~$x$.
The two sheets~$\{n,n+1\}$ support exactly~$\calB(n,\ell)\calA(n,\ell) \otimes \Mat(\ZZ^{\dd-1},p')$
where $\calB(n,\ell)$ is reached by~$\alpha(x)$ for some~$x$
and $\calA(n,\ell)\Mat(\ZZ^{\dd-1},p')$ is reached by~$\phi$.
The sheet~$\{n+2\}$ is reached by~$\beta_\mathrm{right}(z^{\{n+1\}})$
and~$\beta_\mathrm{left}(y^{\{n+3\}})$.
The sheet~$\{j\}$ with $j \ge n+3$ is reached by~$\beta_\mathrm{right}(z^{\{j-1\}})$
and~$\beta_\mathrm{left}(y^{\{j+1\}})$.
\end{proof}

\begin{table}[t]
\caption{
A blending QCA from a QCA with Brauer trivial boundary algebra into a bilayer shift QCA.
Only the first coordinate is shown to specify a subset of~$\ZZ^\dd$.
Symbols~$x \in \Mat( (-\infty,n] \times \ZZ^{\dd-1}, p)$, 
$y \in \Mat(\ZZ^{\dd-1}, q_\mathrm{left} )$,
and $z \in \Mat( \ZZ^{\dd-1}, q_\mathrm{right} )$ denote arbitrary elements.
Exponents of $x,y,z$ denote the sheet on which the element is supported.
The image~$\phi(y)$ is supported on~$S(n,1) = \{n, n+1\} \times \ZZ^{\dd-1}$.
This table complements the text of the proof for~\ref{lem:TrivialBoundaryAlgebraImpliesBlendedTrivialQCA}.
}
\begin{tabular}{r|c|c|c|c}
\hline
\hline
part of lattice & $(-\infty, n]$ & $\{n+1\}$ & $\{n+2\}$ & $[n+3,\infty)$ \\
\hline
local dimension & $p$ & $p(n+1, \cdot) p' = q_\mathrm{right}$ & $q_\mathrm{left} q_\mathrm{right}$ & $1 \times (q_\mathrm{left} q_\mathrm{right})$ \\
\hline
domain element & $x$ & $z^{\{n+1\}}$ & $y^{\{n+2\}} z^{\{n+2\}}$ & $y^{\{j\}} z^{\{j\}}$ \\
\hline
image & $\alpha(x)$ & $z^{\{n+2\}}$ & $\phi(y)^{[n,n+1]} z^{\{n+3\}}$ & $y^{\{j-1\}} z^{\{j+1\}}$ \\
\hline
image supported on & $(-\infty,n+1]$ & $\{n+2\}$ & $[n,n+1] \sqcup \{n+3\}$ & $[n+2,\infty)$ \\
\hline
\hline
\end{tabular}
\label{tb:blending}
\end{table}

\begin{theorem}\label{thm:main}
The abelian group of all blending equivalence classes of $\dd$-dimensional QCA
(blending along a fixed axis) 
modulo those represented by shift QCA,
is isomorphic to the Brauer group of $(\dd-1)$-dimensional invertible subalgebras.
\end{theorem}

\begin{proof}
The isomorphism is given by the correspondence 
from invertible subalgebra to QCA in~\ref{lem:InvertibleSubalgebraToQCA}.
This correspondence preserves tensor product operation clearly.
We have to show that the correspondence is well-defined for Brauer equivalent classes.
Any Brauer trivial invertible subalgebra is mapped to a QCA whose boundary algebra
is Brauer trivial. By~\ref{lem:TrivialBoundaryAlgebraImpliesBlendedTrivialQCA},
we know the resulting QCA blends into a shift QCA,
which is deemed trivial.
Therefore, we have a group homomorphism from Brauer classes to blending equivalence classes.

Taking the boundary algebra (defined in~\ref{lem:BoundaryAlgebra})
gives an invertible subalgebra by~\ref{lem:BoundaryAlgebraIsInvertible},
whose Brauer class is invariant under blending by~\ref{lem:BoundaryAlgebraIndependentOfCut}.
This is a well-defined group homomorphism from blending equivalence classes
to Brauer classes, and is inverse to the previous group homomorphism due to~\ref{lem:InvertibleSubalgebraToQCA}.
\end{proof}

\section{Invertible translation-invariant Pauli subalgebras}\label{sec:Pauli}

We have established tight connection between QCA and invertible subalgebras,
but have not presented any example of invertible subalgebras other than trivial ones.
Perhaps the easiest instances to write down
are those that are translation invariant
with an explicit set of generators.
This section contains a sharp criterion (\ref{lem:PauliInvertibilityCriterion} below)
to determine invertibility of a translation-invariant 
$*$-subalgebra of $\Mat(\ZZ^\dd,p = \text{const.})$
generated by generalized Pauli operators (Weyl operators).
By a (generalized) {\bf Pauli operator} we mean any product
of~$X = \sum_{j\in \ZZ/p\ZZ} \ket{j+1} \bra j$
and~$Z = \sum_{j \in \ZZ/p\ZZ} \exp(2\pi \ii j / p) \ket j \bra j$
and a phase factor of unit magnitude,
or a finite tensor product thereof.
Note that $X Z = e^{-2\pi\ii / p} Z X$.
In this section, every local dimension assignment is a constant function,
and the letter~$p$ will denote a prime number.

\begin{lemma}\label{lem:pauliDecomposition}
Suppose an invertible subalgebra~$\calA \subseteq \Mat(\Lambda)$ is generated by Pauli operators.
Then, the commutant~$\calB = \calA'$ is also generated by Pauli operators
and any Pauli operator~$o \in \Mat(\Lambda)$ is decomposed as~$o = ab$
where~$a \in \calA$ and~$b \in \calB$ are both Pauli operators 
supported within $\ell$-neighborhood of~$o$.
The operators~$a$ and~$b$ are unique up to a phase factor of unit magnitude.
\end{lemma}
\begin{proof}
Let $b \in \calB$ be an arbitrary finitely supported operator.
Since $b \in\Mat(\ZZ^\dd)$, we may write~$b = \sum_j b_j Q_j$ for~$b_j \in \CC$ 
where~$Q_j$ are Pauli operators that are orthonormal under Hilbert--Schmidt inner product.
This expansion is unique.
Since~$b$ commutes with every Pauli generator~$P$ of~$\calA$,
we have $0 = [P,b] = \sum_j b_j [P,Q_j] = \sum_{j} b_j(1-\omega_j) P Q_j$
where $\omega_j$ is some $p$-th root of unity which may be~$1$.
The set $\{PQ_j\}$ is still Hilbert--Schmidt orthonormal since $P$ is unitary,
so we have $b_j(1-\omega_j) = 0$ for all~$j$.
This means that the Pauli operator expansion of~$b$ contains only Pauli operators~$Q_j$
that actually commute with~$P$.
Since $b$ was arbitrary, we conclude that $\calB$ is generated by Pauli operators.

Now, by definition of invertible subalgebras,
we decompose a Pauli operator~$o \in \Mat(\ZZ^\dd)$ 
as~$o = \sum_j o_j P_j Q_j$ with $o_j \in \CC$
where $P_j \in\calA$ and $Q_j \in \calB$ are Pauli operators supported 
within the $\ell$-neighborhood of~$\Supp(o)$.
Since Pauli operators are a Hilbert--Schmidt orthonormal basis,
we must have only one summand, say $j = k$, in the expansion of~$o$.
Since $\calA \cap \calB = \CC$ by~\ref{lem:InvertibleCentral},
the decomposition $o = o_k P_k Q_k$ gives unique~$P_k$ and~$Q_k$ up to a scalar,
but we see that the scalar must have unit magnitude.
\end{proof}

\begin{example}(central simple but not invertible $*$-subalgebra)\label{ex:1dXZ}
Consider a $*$-subalgebra~$\calA \subseteq \Mat(\ZZ,p = 2)$ 
generated by $X(j)\otimes Z(j+1)$ for all~$j \in \ZZ$
where $X,Z$ are single-qubit Pauli operators 
and $j$ denote the support.
This subalgebra appears in~\cite{FreedmanHastings2019QCA}
as a nonexample of VS subalgebras.

The commutant~$\calB$ of~$\calA$ is generated by $Z(j) \otimes X(j+1)$ for all~$j$.
Indeed, it suffices to consider Pauli operators as they are an operator basis,
and any finitely supported Pauli operator in~$\calB$ 
can be multiplied by $Z(j) \otimes X(j+1)$ for some~$j$
to have a smaller support.
It is easy to see that $\calB$ does not contain any single-site operator.
If $x \in \calA \cap \calB$ is nonscalar,
then there is a rightmost site~$s$ in~$\Supp(x)$,
but in the Pauli operator basis, 
the operator component on~$s$ must commute with both~$X(s)Z(j+1)$ 
and~$Z(j)X(j+1)$, so $s$ is not in the support.
Therefore, the intersection~$\calA\cap\calB$ consists of scalars,
showing that $\calA$ is central.

Consider a group homomorphism~$\varphi$ from the multiplicative group of all Pauli operators in $\Mat(\ZZ,p=2)$
into the additive group~$\ZZ_2$ defined by~$X(j) \mapsto 1$ and $Z(j) \mapsto 1$ for all~$j$.
This is well defined because Pauli operators commute up to a sign.
Then, the Pauli generators of~$\calA$ and~$\calB$ 
are in the kernel of the nonzero map~$\varphi$.
If $\calA$ were invertible, then by~\ref{lem:pauliDecomposition}
every single-site Pauli operator must have a Pauli decomposition into~$\calA\otimes\calB$,
but $\varphi$ gives an obstruction to any such decomposition for Pauli~$X$.
Therefore, $\calA$ is not invertible.

Suppose  $\calJ \subseteq \calA$ is a nonzero two-sided ideal.
Let~$y \in \calJ$ be an element of support the smallest in diameter.
Suppose $\diam(\Supp(y)) \ge 2$.
Then, the operator component at the rightmost site~$s \in \Supp(y)$
must be~$Z(s)$, but mutiplying a generator~$X(s-1)Z(s)$ of~$\calA$
gives an element of~$\calJ$ of strictly smaller support.
Hence, $\diam(\Supp(y)) < 2$.
But the linearly independent generators do not make a single-site operator.
Therefore, $\Supp(y) = \emptyset$ and~$\one \in \calJ$,
proving that~$\calA$ is simple.

If one instead considers a finite lattice under a periodic boundary condition,
then $\calA$ contains a nontrivial central element 
$\cdots \otimes Y \otimes Y \otimes \cdots $ .
\hfill $\square$
\end{example}

Let $\calA \subseteq \Mat(\ZZ^\dd, \ZZ^\dd \to \{p^q\})$ be a $*$-subalgebra,
where $p$ is a prime and $q$ is any positive integer,
generated by Pauli operators~$ P_1,\ldots,P_n$ and their translates.
We identify~\cite{clifQCA,Haah2013} 
the multiplicative group of all finitely supported Pauli operators
modulo phase factors
with the additive group~$\mathbb P$ of column matrices of length~$2q$ 
over a commutative Laurent polynomial ring
\begin{align}
	\polyring = \FF_p[x_1^\pm,\ldots,x_\dd^\pm]
\end{align}
with coefficients in the finite field~$\FF_p$.
The Laurent polynomial ring has an $\FF_p$-linear involution, denoted by bar,
that sends each variable~$x_i$ to~$\bar x_i = x_i^{-1}$.
Let $v_j$ be such a column matrix over~$\polyring$
corresponding to~$P_j$ for $j = 1, \ldots,n$.
We form an $n \times n$ {\bf commutation relation matrix}~$\Xi$ over~$\polyring$ 
by
\begin{align}
	\Xi_{ij} = \overline{v_i}^T \lambda_q v_j
\end{align}
where $\lambda_q$ is a nondegenerate antihermitian matrix
that is~$\begin{pmatrix} 0 & \id_q \\ -\id_q & 0 \end{pmatrix}$
if we use a basis for~$\mathbb P = \polyring^{2q}$ such that 
the upper $q$ components represent the $X$-part and the lower the $Z$-part;
see~\cite{Haah2013}.

\begin{proposition}\label{lem:PauliInvertibilityCriterion}
Let $\calA \subseteq \Mat(\ZZ^\dd, \ZZ^\dd \to \{p^q\})$ be a $*$-subalgebra
generated by a translation-invariant set of Pauli operators, 
with a commutation relation matrix~$\Xi$,
where $p$ is a prime and $q \ge 1$ an integer.
Then, $\calA$ is invertible 
if and only if the smallest nonzero determinantal ideal of~$\Xi$ is unit.
\end{proposition}
\begin{proof}
Suppose $\calA$ is invertible.
Then, the $(\dd+1)$-dimensional QCA constructed in~\ref{lem:InvertibleSubalgebraToQCA}
maps a Pauli operator to a Pauli operator and is translation invariant.
Hence, the QCA is Clifford~\cite{clifQCA,clifqca1}.
Lemma~\cite[III.2]{clifqca1} shows that the smallest nonzero determinantal ideal of~$\Xi$ is unit.

Conversely, suppose that the smallest nonzero determinantal ideal of~$\Xi$ is unit.
Then,~\cite[III.3]{clifqca1} says that there is an invertible $n \times n$ matrix~$G$ 
over~$\polyring$ such that
\begin{align}
	\bar G^T \Xi G = \Theta \oplus 0
\end{align}
where $\Theta$ is $m \times m$ and has unit determinant, where $m \le n$.
This equation can be rewritten as $\bar G^T \bar V^T \lambda_q V G = \Theta \oplus 0$
where the columns of~$V$ represent the initial set of generators (up to translations)
for~$\calA$.
A column~$v$ of~$V G$ such that $\bar v^T \lambda_q V G = 0$ represents
a central element of~$\calA$, 
which must be trivial by~\ref{lem:InvertibleCentral}.
Hence, $v = 0$.
This means that the right $n-m$ columns of~$G$ 
encode multiplicative relations 
among the initial choice of generators for~$\calA$.
Thus, the left $m$ columns of the product~$V G$, 
which we collect in an $2q \times m$ matrix~$A$ over~$\polyring$,
represents another set of Pauli generators for~$\calA$,
and 
\begin{align}
\Theta = \bar A^T \lambda_q A
\end{align}
is an invertible matrix over~$\polyring$.
Now, consider a projection acting on~$\polyring^{2q}$:
\begin{align}
	\Pi = A \Theta^{-1} \bar A^T \lambda_q = \Pi^2 .
\end{align}
It is readily checked that $\Pi A = A$ and $\bar A^T \lambda_q(\id - \Pi) = 0$.
So, $\polyring^{2q} = \Pi \polyring^{2q} \oplus (\id - \Pi) \polyring^{2q} = A \perp (\id - \Pi) \polyring^{2q}$,
where $\perp$ is with respect to~$\lambda_q$.
Interpreting column matrices (``vectors'') over~$\polyring$ as Pauli operators,
we see that $(\id - \Pi)\polyring^{2q}$ describes the commutant~$\calB$ of~$\calA$
and the projections $\Pi$ and $\id - \Pi$ give the desired decomposition
for~$\calA$ to be invertible.
The locality requirement is fulfilled 
since $\Pi$ and $\id - \Pi$ have entries that are
fixed polynomials with finitely many terms.
The exponents in these projectors give the spread parameter.
\end{proof}

For example, the subalgebra in~\ref{ex:1dXZ}
has $1 \times 1$ commutation relation matrix
\begin{align}
\overline{\begin{pmatrix} 1 \\ x_1 \end{pmatrix}}^T 
\begin{pmatrix}
0 & 1 \\ -1 & 0
\end{pmatrix}
\begin{pmatrix} 1 \\ x_1 \end{pmatrix}
= x_1 - x_1^{-1},
\end{align}
which has smallest nonzero determinantal ideal $(x_1 - x_1^{-1}) \subset \FF_2[x_1^\pm]$
that is not unit.

\begin{remark}\label{rem:XiToInvertible}
Given any invertible antihermitian $n \times n$ matrix~$\Xi = M - \bar M ^T$ for some~$M$
over~$\polyring = \FF_p[x_1^\pm,\ldots,x_\dd^\pm]$
there exists an invertible Pauli $*$-subalgebra on~$\Mat(\ZZ^\dd)$.
The construction is from an observation that
\begin{align}
	\Xi = 
	\begin{pmatrix} \id & \bar M^T \end{pmatrix} 
	\begin{pmatrix} 0 & \id \\ -\id & 0 \end{pmatrix}
	\begin{pmatrix} \id \\ M	\end{pmatrix}.
\end{align}
where the $2n \times n$ matrix conjugating $\lambda_n$
contains columns each of which represents a generator of a Pauli $*$-subalgebra.
If $\half \in \polyring$,
one can use
\begin{align}
	\Xi = 
	\begin{pmatrix} \id & \half \overline \Xi^T \end{pmatrix} 
	\begin{pmatrix} 0 & \id \\ -\id & 0 \end{pmatrix}
	\begin{pmatrix} \id \\ \half\Xi	\end{pmatrix}.
\end{align}
This will be used in the next section.
\hfill $\square$
\end{remark}

\section{Topological order within invertible subalgebras}\label{sec:exampleTopologicalOrder}

\begin{figure}[t]
\centering
\includegraphics[width=0.7\textwidth, trim={0ex 70ex 90ex 0ex}, clip]{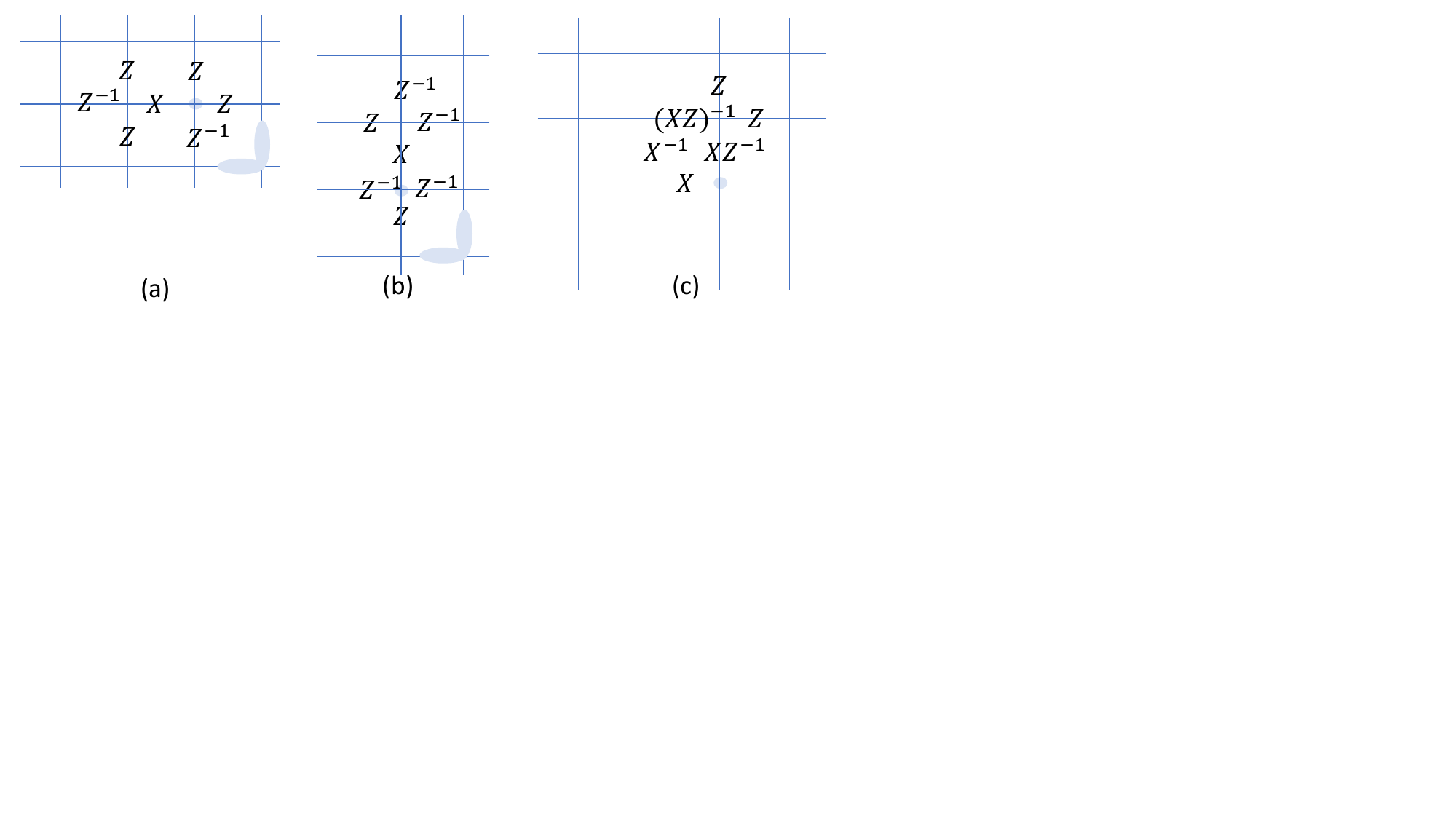}
\caption{An example of two-dimensional invertible subalgebra.
It is generated by the operators in~($\mathsf a$) and~($\mathsf b$) and their translates.
To help recognizing this drawing from the Laurent polynomial description, 
we have drawn a unit cell on the lower right corner
and marked the origin of the lattice by a pale dot.
($\mathsf c$) shows a Hamiltonian term that belongs to this invertible subalgebra,
defining a commuting Pauli Hamiltonian.
This term is obtained by multiplying 
one~($\mathsf a$), one~($\mathsf b$),
and their inverses.
The excitations of this Hamiltonian realize an abelian anyon theory with three anyons 
(two nontrivial, one vacuum)
whose topological spins~$\theta_j$ 
are $\{ 1, e^{2\pi \ii / 3},e^{2\pi \ii / 3} \}$.
This anyon theory is chiral: 
$(\sum_j d_j^2)^{-1/2}\sum_j d_j^2 \theta_j = e^{2\pi \ii c_-/8} = \ii$
where $d_j = 1$.
This anyon theory together with its time reversal conjugate is
just the toric code over qutrits.
}
\label{fig:hopping}
\end{figure}

In this section we discuss an example inveritible subalgebra in two dimensions.
The example is presented in~\cite[\S IV]{clifqca1}, though somewhat less explicitly.
We will see that this invertible subalgebra admits a \emph{commuting} Hamiltonian 
whose excitations realize a \emph{chiral} anyon theory that is not a quantum double.
The anyon theory is an abelian one with the fusion group~$\ZZ_3$
and topological spin~$e^{2\pi\ii k^2 / 3}$ for~$k \in \ZZ_3$.
The algebraic data of anyons give the value of chiral central charge~$c_- = 2 \bmod 8$~\cite{FroehlichGabbiani1990braid}.
Kitaev~\cite[App.~D]{Kitaev_2005} has developed some theory
to write a formula that expresses this chiral central charge by bulk Hamiltonian terms,
and the formula is given by a sum of expectation values of certain commutators of bulk terms.
For commuting Hamiltonians the commutators vanish tautologically,
and the formula evaluates to zero.

Perhaps we should explain our example here from an abstract angle.
When we construct~$\Mat$ as a mathematical model of discrete physical degrees of freedom,
it is our choice that the operator algebra of two or more degrees of freedom
is given by the tensor product of those of individual ones.
More generally, we may consider a discrete net of operator algebras,
which is a collection of unital algebras $\calA_S$, one for each region $S \subseteq \ZZ^\dd$,
such 
that any inclusion map of regions induces a homomorphism of algebras,
that algebras over disjoint regions are mapped to mutually commuting algebras,
and that the union~$\bigcup_F \calA_F$ of all algebras over finite regions~$F$
is dense in~$\calA = \calA_{\ZZ^\dd}$ under an appropriate sense.
One can add extra assumptions such as norm-completeness.
For fermionic systems we should instead use $\ZZ_2$-graded commutation relations.
This notion of discrete nets of operator algebras 
allows us to speak of local operators without explicitly specifying tensor products.
In a net of operator algebras over a many-body system,
we consider a Hamiltonian by choosing certain
local hermitian operators representing interactions.
Then we consider eigenstates of the Hamiltonian,
which are linear functionals on observables.
We also consider superselection sectors.

In addition to the familiar~$\Mat$,
an invertible subalgebra~$\calA \subseteq \Mat$ 
is another example of a discrete net of operator algebras:
there is a notion of support for every element of~$\calA$ inherited from~$\Mat$,
so we consider an algebra~$\calA_S$ consisting of all those supported on any given region~$S$,
and $\calA = \calA_{\ZZ^\dd}$ is generated by those local subalgebras; see~\ref{lem:csa}.
In a net of operator algebras, we may choose a commuting local Hamiltonian
and analyze the topological quasi-particle content.
This is the question that we have alluded by the title of this section.
So, the apparent paradox that a commuting Hamiltonian realizes a chiral theory
appears to be rooted in the unusual net of operator algebras;
however, it would be worth figuring out
what step in the derivation of Kitaev's formula~\cite{Kitaev_2005} relies on
a particular choice of a net of operator algebras.

We consider~$\Mat(\ZZ^2, 9)$, 
the local operator algebra on square lattice with $\CC^3$ on every edge.
Starting with an antihermitian matrix~\cite[\S IV]{clifqca1}
\begin{align}
	\Xi = 
	\left(
\begin{array}{cc}
 \frac{1}{x}- x & x +x y -y +1 \\
 -\frac{1}{x}+\frac{1}{y}-\frac{1}{x y}-1 &  y-\frac{1}{y} \\
\end{array}
\right) = - \overline \Xi ^T, \quad \det \Xi = 4, \label{eq:ExampleXi}
\end{align}
over $\polyring = \FF_3[x^\pm,y^\pm]$,
we consider a $*$-subalgebra~$\calA$ by the prescription of~\ref{rem:XiToInvertible}.
We draw two generators in Figure~\ref{fig:hopping};
the $*$-subalgebra~$\calA$ is generated by all translates of 
what is drawn in Figure~\ref{fig:hopping}($\mathsf a, \mathsf b$).
Consider a translation-invariant local Hamiltonian~$H$ 
whose terms are Pauli operators in~$\calA$,
defined as
\begin{align}
	H = - \sum_{s \in \ZZ^2} P_s + P_s^\dag, \qquad P_s = \text{Figure~\ref{fig:hopping}}(\mathsf c) .\label{eq:HamiltonianPs}
\end{align}

\begin{proposition}\label{thm:ExampleInv}
All of the following hold for the $*$-subalgebra~$\calA$.
\begin{enumerate}
	\item[(i)] $\calA$ is invertible. 
	There is a locality-preserving $*$-isomorphism $\calA^{\otimes 4} \cong \Mat(\ZZ^2,3^4)$,
	showing that $\calA^{\otimes 4}$ is Brauer trivial.
	\item[(ii)] The commutant~$\calB$ of~$\calA$ within~$\Mat(\ZZ^2,9)$ is equal to the complex conjugate of~$\calA$ in the basis where $Z$ is diagonal and $X$ is real.
	\item[(iii)] Let~$\mathcal H \subset \calA$ be the commutative $*$-subalgebra
generated by~$\{P_s | s \in \ZZ^2\}$.
The commutant of~$\mathcal H$ within~$\calA$ is equal to~$\mathcal H$.
	\item[(iv)] Let $h$ be a finite product of~$P_s, P_s^\dag$ 
	over a nonempty finite set of sites~$s$
	such that there is one factor~$P_s$ or~$P_s^\dag$ for each~$s$.
Then, $h \notin \CC\one$.
	\item[(v)] Given any two sites~$s_1 \neq s_2$ of~$\ZZ^2$,
there exists a finite product~$y$ of Pauli generators of~$\calA$
such that $y$ commutes with all~$P_s$ with~$s \neq s_1, s_2$ 
but does not with~$P_{s_1}$ and~$P_{s_2}$.
\end{enumerate}
\end{proposition}

\begin{proof}
(i)
The commutation relation matrix of~$\calA$ is in~\eqref{eq:ExampleXi}.
Apply~\ref{lem:PauliInvertibilityCriterion}.
The claim that $\calA^{\otimes 4} \cong \Mat(\ZZ^2,81)$ is basically 
proved in~\cite[III.10]{clifqca1},
which implies that $\Xi^{\oplus 4}$ is congruent to~$\lambda_4 =
\begin{pmatrix} 0 & \id_4 \\ -\id_4 & 0 \end{pmatrix}$
over~$\polyring$.
Recall that the commutation relation matrix~$\Xi$ is equal
to~$\bar V^T \lambda_2 V$ for a $4 \times 1$ matrix~$V$ over~$\polyring$
where the columns of~$V$ represent the Pauli generators of~$\calA$.
The congruence gives an invertible matrix~$E$ over~$\polyring$
such that $\lambda_4 = \bar E^T \Xi^{\oplus 4} E 
= \bar E^T (\bar V^T)^{\oplus 4} \lambda_2^{\oplus 4} V^{\oplus 4} E$.
Let $F$ be a permutation matrix such that $F^T \lambda_8 F = \lambda_2^{\oplus 4}$.
The map $F V^{\oplus 4} E : \polyring^8 \to \polyring^{16}$
represents a locality-preserving 
Clifford $*$-map from~$\Mat(\ZZ^2, 3^4)$ onto~$\calA^{\otimes 4} \subset \Mat(\ZZ^2, 3^8)$.

(ii)
The commutant~$\calB$ has Pauli generators corresponding to an $\polyring$-submodule
\begin{align}
	\ker\left( 
	\begin{pmatrix} \id & \half \overline \Xi^T \end{pmatrix}
	\begin{pmatrix} 0 & \id \\ -\id & 0 \end{pmatrix}
	\right)
	=
	\polyring \begin{pmatrix} \id \\ -\half \Xi \end{pmatrix}
	\label{eq:submoduleP}
\end{align}
This is precisely the complex conjugate of~$\calA$
since the minus sign in the lower half block corresponds to $Z \mapsto Z^{-1}$.

(iii)
For each $s \in \ZZ^2$, 
let $\varphi_s : x \mapsto \frac 1 3 (x + P_s x P_s^\dag + P_s^2 x (P_s^2)^\dag )$
be a projection.
Since $P_s$'s commute, 
we see $\varphi_s \varphi_{s'} = \varphi_{s'} \varphi_s$ for all~$s,s' \in \ZZ^2$.
Note that for any Pauli operator~$Q$,
\begin{align}
	\varphi_s(Q) = \begin{cases} Q & \text{if } [P_s, Q] = 0, \\ 0 & \text{otherwise}. \end{cases}
\end{align}
Suppose $y \in \calA$ commutes with every element of~$\mathcal H$.
We must have $\varphi_s(y) = y$ for all~$s$.
Since~$P_s$'s generate~$\calA$,
we have a finite linear combination~$y = y_0 \one + \sum_s y_s P_s + y'_s P_s^2$
where $y_0, y_s, y_s' \in \CC$.
Applying $\varphi_s$ for various $s$,
we see that any nonzero summand in the expansion must be commuting with every~$P_s$.

Therefore, it suffices to check that any finitely supported Pauli operator~$Q$ of~$\calA$ 
that commutes with all~$P_s$ is a finite product of~$P_s$'s.
To this end, we use the Laurent polynomial representation of Pauli operators.
The operator $Q$ corresponds to a column matrix $\begin{pmatrix} \id \\ \half \Xi \end{pmatrix} \begin{pmatrix} a \\ b \end{pmatrix}$ for some $a,b \in \polyring = \FF_3[x^\pm,y^\pm]$.
It is easy to check that the Hamiltonian term~$P_s$ corresponds to
$\begin{pmatrix} \id \\ \half \Xi \end{pmatrix} \begin{pmatrix} 1-y \\ 1 - x^{-1} \end{pmatrix}$.
Their commutativity is expressed as
\begin{align}
0 = \begin{pmatrix} 1-y^{-1} & 1 - x \end{pmatrix}
\begin{pmatrix} \id & \half \overline \Xi^T \end{pmatrix} 
	\begin{pmatrix} 0 & \id \\ -\id & 0 \end{pmatrix}
	\begin{pmatrix} \id \\ \half\Xi	\end{pmatrix}
\begin{pmatrix} a \\ b \end{pmatrix}
=
y^{-1}( -1+ x^{-1} )a + (-1 + y^{-1})b .
\end{align}
The coefficients~$(-1+x^{-1})$ and~$(-1+y^{-1})$ are coprime in~$\polyring$,
so we must have
\begin{align}
\begin{pmatrix} a \\ b \end{pmatrix} = r \begin{pmatrix} 1-y \\ 1-x^{-1} \end{pmatrix}
\end{align}
for some~$r \in \polyring$,
which precisely expresses the condition that $Q \in \mathcal H$.

(iv)
The claim amounts to the injectivity of the map~$\polyring \to \polyring^4$ given by the matrix
$\begin{pmatrix} \id \\ \half \Xi \end{pmatrix} \begin{pmatrix} 1-y \\ 1 - x^{-1} \end{pmatrix}$.
This is a nonzero column matrix over a ring~$\polyring$ without zerodivisors, so the kernel is zero.

(v)
The claim is easily proved by considering which Hamiltonian terms are noncommuting
with ``short string operators,'' the generators drawn in Figure~\ref{fig:hopping}.
A straightforward calculation shows
\begin{align}
\begin{pmatrix} 1-y^{-1} & 1 - x \end{pmatrix}
\begin{pmatrix} \id & \half \overline \Xi^T \end{pmatrix} 
	\begin{pmatrix} 0 & \id \\ -\id & 0 \end{pmatrix}
	\begin{pmatrix} \id \\ \half\Xi	\end{pmatrix}
\begin{pmatrix} 1 & 0 \\ 0 & 1 \end{pmatrix}
= \begin{pmatrix}
y^{-1}(-1 + x^{-1}) & -1 + y^{-1}
\end{pmatrix}.
\end{align}
The interpretation is that the generator~($\mathsf a$)
commutes with all the Hamiltonian terms except the two at~$(0,-1) \in \ZZ^2$ and~$(-1,-1)$,
and ($\mathsf b$) commutes with all terms except the two at~$(0,0)$ and~$(0,-1)$.
Hence, some juxtaposition of the short string operators commutes with all terms except those 
at~$s_1,s_2$ for any given $s_1, s_2 \in \ZZ^2$.
\end{proof}

We are now ready to say that the Hamiltonian realizes a $\ZZ_3$ anyon theory
by the results of~\ref{thm:ExampleInv}.
The complex conjugate~$\bar H$ of~$H$ (in the basis where $Z$ is diagonal and $X$ is real)
belongs to the commutant of~$\calA$ by~(ii).
Hence, a Hamiltonian~$H + \bar H$ on the full local operator algebra,
is a sum of noninteracting two theories.
It is easy to see that the terms of $H + \bar H$ generate multiplicatively 
the terms of the toric code Hamiltonian~\cite{Kitaev_2003} with qutrits.
The statement~(iv) means that the Hamiltonian~$H + \bar H$ is not frustrated,
and there is a ground state, 
a nonnegative normalized linear functional on the algebra of observables,
with respect to which all~$P_s$ are valued~$+1$.
In particular, $H$ has a ground state with which all~$P_s$ assume~$+1$.
By~(iii), the ground state on~$\calA$ is unique on our infinite lattice.
By~(v) an excited state 
with an even number of violated terms (not counting $P_s^\dagger$ separately) 
is an admissible state,
and in particular it is possible to isolate an excitation at the origin 
with its antiexcitation far from the origin.
These isolated excitations obey the $\ZZ_3$ fusion rule because $P_s^3 = \one$.
Considering an exchange process along three long line segments with one common point at the origin~\cite{LevinWen2003Fermions,clifqca1},
one can read off the topological spin of an isolated excitation.
We have performed detailed calculation in~\cite{clifqca1},
and the result is $e^{2\pi \ii /3 }$.
It is important here that the hopping operators (the short and long string operators)
and the Hamiltonian are entirely built from~$\calA$.
In this sense, our Hamiltonian realizes the chiral $\ZZ_3$ anyon theory within~$\calA$.

If $\calA$ was Brauer trivial, 
then its commutant~$\calB$ within~$\Mat(\ZZ^2)$ would also be Brauer trivial.
In any Brauer trivial invertible subalgebra,
we can define a trivial Hamiltonian
that is a sum of single-site operators, after a possible stabilization.
The isomorphism that gives Brauer triviality 
is locality-preserving by definition,
so it may turn the trivial Hamiltonian into a less trivial Hamiltonian.
However, such a locality-preserving map
cannot change the emergent particle content.
Overall, we would have a commuting Hamiltonian on~$\Mat(\ZZ^2)$
that realizes the chiral $\ZZ_3$ anyon theory.
Believing this is not possible,
we have to conclude that $\calA$ is Brauer nontrivial.

All examples of invertible subalgebras in~\cite{clifqca1}
and that of~\cite{nta3},
though we did not call them invertible,
exhibit similar phenomena
that some commuting Hamiltonian realizes an abelian anyon theory 
whose topological spin is a quadratic form of a nontrivial Witt class.
This is contrasted to a folklore that any commuting Hamiltonian in two dimensions
is stably quantum circuit equivalent to a Levin--Wen model~\cite{Levin_2005}.
Recently discovered QCA in three dimensions~\cite{Shirley2022} and in four dimensions~\cite{Chen2021}
give invertible subalgebras in two and three dimensions, respectively.
The two-dimensional invertible subalgebra should admit a commuting Hamiltonian,
realizing the semion theory in which $c_- = 1 \bmod 8$.
The three-dimensional invertible subalgebra should 
admit ``FcFl'' model~\cite{Chen2021,Fidkowski2021},
which is supposed not to exist in~$\Mat(\ZZ^3)$.

\begin{remark}\label{rem:conjecture}
For two-dimensional invertible subalgebras,
we wonder if the Brauer group coincides with the Witt group of modular tensor 
categories (MTC)~\cite{Davydov2010,Davydov2011}.
This expectation is backed by the following heuristic argument.

First, any two-dimensional invertible subalgebra
would admit a commuting Hamiltonian that realizes some anyon theory,
corresponding to a MTC.
So, we have a map from from an invertible subalgebra into the Witt group.
To argue that this map is well defined on the Brauer class of an invertible subalgebra,
we consider the three-dimensional QCA restricted on a finitely thick slab that is two-dimensional,
where the QCA is constructed as in~\ref{lem:InvertibleSubalgebraToQCA}.
We have the top boundary hosting an anyon theory, which we fix,
and the bottom boundary may host another anyon theory.
But the whole system is basically two-dimensional,
and all Hamiltonians are commuting.
According to the folklore that Levin--Wen model is all one may realize 
by a commuting Hamiltonian on~$\Mat(\ZZ^2)$,
which means that the whole theory on the slab must be Witt trivial,
the two theories at the top and the bottom must represent inverse Witt classes of each other.
In other words, the Witt class of the anyon theory at the bottom
is determined by that at the top,
and therefore cannot change under any stable equivalence of invertible subalgebras.

Second, for every MTC
there is a Walker--Wang model~\cite{Walker_2011}
whose bulk theory lacks any nontrivial topological excitation,
but the input MTC is realized at a boundary,
all by some commuting local Hamiltonian.
The bulk would be disentangled by a three-dimensional QCA,
from which we take a boundary invertible subalgebra.
This gives a map in the converse direction.
To argue that this is well defined on the Witt class of a MTC,
we observe that the Walker--Wang model of a Witt trivial MTC
should admit a trivial boundary,
then corresponding disentangling QCA would blend into the identity,
and the boundary invertible subalgebra must be Brauer trivial.

The first and the second maps are inverses of each other.
\hfill $\square$
\end{remark}

\section{Discussion}\label{sec:discussion}

We have introduced a notion of invertible subalgebras (\ref{def:invertible}),
which completely characterizes boundary algebras of a QCA~\cite{GNVW,FreedmanHastings2019QCA} 
in one dimension higher (\ref{lem:BoundaryAlgebraIsInvertible}, \ref{lem:InvertibleSubalgebraToQCA}).
We have shown that blending equivalence classes of $\dd$-dimensional QCA modulo shifts
are in one-to-one correspondence with locality-preserving stable isomorphism 
classes of $(\dd-1)$-dimensional invertible subalgebras (\ref{thm:main}).
On finite systems, either open or periodic boundary conditions,
the notion of invertible subalgebras is equivalent 
to that of ``visibly simple'' subalgebras (\ref{lem:InvertibleCentral}, \ref{prop:VS}),
which was introduced in~\cite{FreedmanHastings2019QCA}.

Though we have defined the Brauer group of invertible subalgebras given a spatial dimension,
we give no results how large or small this Brauer group is in general.
Every invertible subalgebra in zero dimension is trivially Brauer trivial.
Every invertible subalgebra in one dimension 
is shown to be Brauer trivial (\ref{thm:BrauerGroupIsTrivialIn1D}).
Restricting the class of invertible subalgebras 
to those generated by prime dimensional Pauli operators,
and restricting stable isomorphism maps between those Pauli invertible subalgebras
to ones that map Pauli operators to Pauli operators (Clifford),
we understand this Pauli/Clifford Brauer group 
in all spatial dimensions~\cite{clifQCAclassification}:
there is a periodicity in spatial dimension with period either~2 or~4,
and nonzero groups are given by classical Witt groups of finite fields, 
depending on the local prime dimension.
This result is in line with the conjecture in~\ref{rem:conjecture}.
Finding invariants of invertible subalgebras, not necessarily Pauli,
remains an obvious, critical problem.

\begin{figure}
\centering
\includegraphics[width=0.8\textwidth, trim={0ex 60ex 80ex 0ex}, clip]{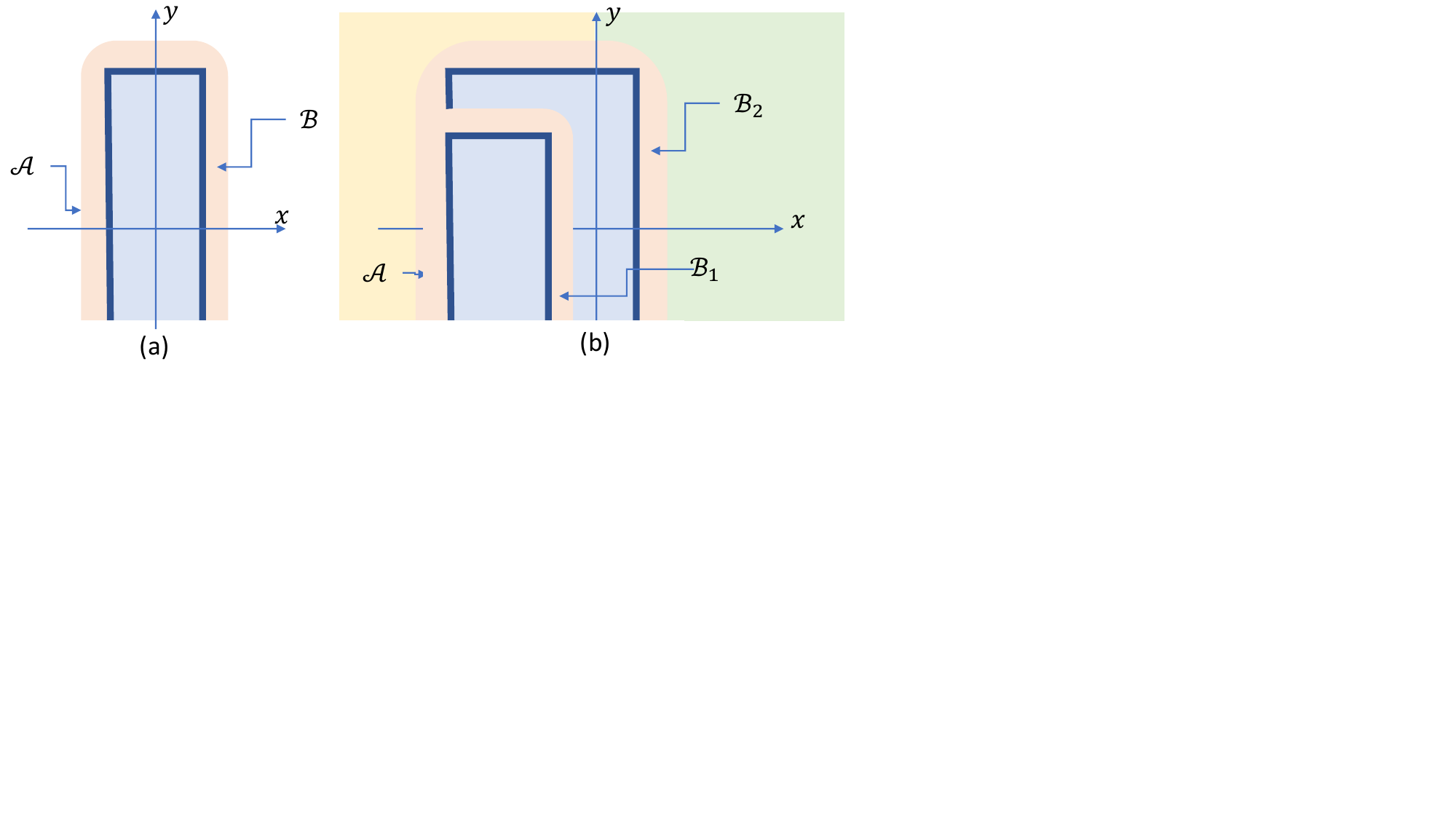}
\caption{Blending of two invertible subalgebras.
($\mathsf a$) As long as the horizontal width of the blue region 
is much larger than the spread of the QCA,
the boundary algebra is well defined on the folded pink region.
The decomposition map of~\ref{lem:TensorIso}, that is the inverse of~$\mul$, 
for this folded boundary algebra
agrees on the right boundary with that of the original, unfolded, boundary algebra~$\calB$
on the positive $x$-axis,
and agrees on the left boundary with that of the original, unfolded, boundary algebra~$\calA$
on the negative $x$-axis.
Hence, $\calA$ and $\calB$ blend,
but due to the folding we have to invert the $y$-coordinate on one of the algebras
to put them on a common flat space.
The right pane~($\mathsf b$) 
shows that the stable equivalence of two invertible subalgebras
implies the blending equivalence.
}
\label{fig:blending}
\end{figure}

In view of~\ref{lem:TensorIso}, 
we may introduce blending equivalence of invertible subalgebras.
Since an invertible subalgebra~$\calA$ is determined 
by a locality-preserving map~$\Mat \to \calA \otimes \calB$
(the inverse of~$\mul$)
where $\calB$ is the commutant of~$\calA$ within~$\Mat$,
we consider interpolation between two different such maps 
across two infinite halves of the lattice.
An example of invertible subalgebra blending appears 
when we consider the boundary algebra of a QCA with spatial boundary folding.
It is not difficult to see in Figure~\ref{fig:blending}($\mathsf a$)
that an invertible subalgebra~$\calA$ 
blends into its commutant after one-coordinate inversion;
see~\cite{FHH2019} for some discussion of one-coordinate inversion.
This blending gives an equivalence relation.
Note that
	\emph{any two stably equivalent invertible subalgebras always blend.}
	Indeed, by~\ref{thm:main}, these invertible subalgebras~$\calB_1$ and~$\calB_2$ 
	give two QCA that blend in one dimension higher,
	and a blending QCA between the two QCA shows that $\calB_1$ blends 
	into the one-coordinate-inversion of~$\calA$ as in Figure~\ref{fig:blending}($\mathsf b$),
	which blends into $\calB_2$.
Considering higher dimensions,
we see that boundary algebras on different axes blend into each other
after an appropriate coordinate transformation.

Beyond the problem to find invariants,
there are various questions left unresolved.
For example,
we do not know if the notion of visibly simple subalgebra 
is equivalent to that of invertible subalgebra on infinite dimensional local operator algebras.
Another question is
whether the image of a locality-preserving map from an invertible subalgebra is always invertible.
This question may be thought of as a domain extension problem 
in which one is given a map from an invertible subalgebra of~$\Mat$,
and wishes to extend it to the entire~$\Mat$ where the invertible subalgebra resides.

Finally, an important avenue to explore 
is to consider an approximately locality-preserving setting~\cite{Ranard2022}.
Our invertibility demands that there be a hard bound on 
the support of the decomposition of form~$x = \sum_j a_j b_j$.
In quasi-local operator algebras, which are the norm-completion of local operator algebras,
this is probably too restrictive.
We would want approximately locality-preserving maps in the stable equivalence
of invertible subalgebras.
In particular, 
if we wish to include Hamiltonian evolution on the same footing as QCA,
then the decomposition for invertibility must be relaxed.
Some choice will have to be made for what we mean by approximately locality-preserving.
Colloquially speaking, 
we need to choose how fast the tails of $*$-automorphisms should decay.
The result in the appendix that Hamiltonian time evolution
is a limit of quantum circuits,
motivates us that we might want to 
think of approximately locality-preserving $*$-automorphisms
as limits of QCA.
The main point of~\cite{Ranard2022} is that 
a certain class of those on one-dimensional lattice
are indeed limits of QCA.

\appendix

\section{Topology of linear transformations on local operator algebras}\label{app:topo}

We are going to show that time evolution by a local Hamiltonian for any given evolution time
is a limit of a sequence of QCA, 
each of which is a finite depth quantum circuit.
This formalizes the statement that quantum circuits approximate Hamiltonian time evolution,
which is nontrivial since our operator space is infinite dimensional.
The content of this statement 
may depend significantly on the topology
of the space to which the time evolution and quantum circuits belong.
By choosing a sufficiently fine topology on this space,
we wish to settle on an interesting and fruitful class 
of approximately locality-preserving $*$-automorphisms.
This perspective is perhaps complementary to
the result of~\cite{Ranard2022}.
We will comment on it after~\ref{thm:TimeEvolutionIsLimit} below.

Our exposition will contain some standard notions and results
such as trace, norm, local Hamiltonians, derivations, and time evolution.
These can all be found in a textbook~\cite{BratteliRobinson2},
but we remain as elementary as possible.

We simply write~$\Mat$ instead of~$\Mat(\ZZ^\dd,p)$ when the spatial dimension~$\dd$ 
and local dimension assignment~$p$ are not important.
Nowhere in our analysis will the local dimension appear.
The operator space~$\Mat$ and its completion~$\barMat$ 
are always endowed with the standard norm topology.
Consider a real vector space
\begin{align}
\calL = \{ \varphi : \Mat \to \barMat ~|~ \varphi \text{ is $*$-linear.} \}
\end{align}
No locality-preserving property is enforced,
though we will primarily interested in those with a locality-preserving property.
Every $*$-algebra automorphism of~$\Mat$ is a member of~$\calL$.
As we will see, a derivation $x \mapsto [\ii H,x]$ by a local Hamiltonian~$H$ is
a member of~$\calL$.

\begin{definition}\label{def:dist}
For any $\alpha,\beta \in \calL$, we define
\begin{equation}
\dist(\alpha,\beta) = \sup_{x \in \Mat \setminus \CC\one} \frac{\norm{\alpha(x) - \beta(x)}}{\norm x \cdot \abs{\Supp(x)}}
\end{equation}
where $\abs{\Supp(x)}$ denotes the number of sites in~$\Supp(x)$.
Note that the supremum is taken over all finitely supported elements of~$\Mat$.
\end{definition}

The appearance of $\norm x$ in the denominator is necessary,
because, otherwise, the supremum will be simply infinite as seen by scaling~$x$ by a scalar.
The choice of the other factor~$\abs{\Supp(x)}$ is somewhat arbitrary,
but there is some reason we would want such a factor 
that is a growing function in~$\abs{\Supp(x)}$,
as we will explain in~\ref{rem:reasonSupp} below.

Our metric topology on~$\calL$
sits between the strong topology and the induced norm topology.
The induced norm topology is obtained by removing~$\abs{\Supp(x)}$ factor
in the definition of~$\dist$.
It is equivalently defined as one in which
$\alpha_i$ converges to~$\beta$ if and only if
$\lim_{i \to \infty} \sup_{x \in \Mat(\ZZ^\dd) : \norm{x} = 1} \norm{\alpha_i(x) - \beta(x)} = 0$.
This topology is too fine to be useful;
we will see in~\ref{rem:continuity} below that 
the time evolution operator of
a noninteracting local Hamiltonian is not continuous in time.
In the strong topology,
a sequence of linear transformations~$\alpha_i$ converges to another~$\beta$
if and only if 
$\lim_{i \to \infty} \norm{\alpha_i(x) - \beta(x)} = 0$ for any~$x \in \Mat(\ZZ^\dd)$.
It is pointwise convergence.
The strong topology is a conventional choice 
in a classic book by Bratteli and Robinson~\cite{BratteliRobinson2}
when deriving time evolution by a lattice local Hamiltonian
as a limit of those on finite subsystems of increasing volume.
We will show in~\ref{rem:strongTopology}
that our metric topology is strictly finer than the strong topology.

\begin{proposition}\label{thm:distIsMetric}
$\calL$ is a complete metric space under~$\dist$.
\end{proposition}

\begin{proof}
First, we have to show that $\dist$ is a metric.
It is obvious that $\dist(\alpha,\beta) = \dist(\beta,\alpha)$.
Suppose $\dist(\alpha,\beta) = 0$. Clearly, $\delta = \alpha - \beta = 0$ on~$\Mat \setminus \CC\one$.
If $x \in \Mat \setminus \CC\one$, we have $\norm{\delta(\one)} = \norm{\delta(\one - x) + \delta(x)} \le \norm{\delta(\one - x)} + \norm{\delta(x)} = 0$, so $\delta = 0$ on the entire~$\Mat$.
Triangle inequality follows by
\begin{align}
\frac{\norm{\alpha(x) - \gamma(x)}}{\norm x \cdot \abs{\Supp(x)} }
&\le
\frac{\norm{\alpha(x) - \beta(x)}}{\norm x \cdot \abs{\Supp(x)} }
+
\frac{\norm{\beta(x) - \gamma(x)}}{\norm x \cdot \abs{\Supp(x)} } \nonumber\\
&\le 
\sup_y \frac{\norm{\alpha(y) - \beta(y)}}{\norm y \cdot \abs{\Supp(y)} }
+
\sup_z \frac{\norm{\beta(z) - \gamma(z)}}{\norm z \cdot \abs{\Supp(z)} }\\
&= \dist(\alpha,\beta) + \dist(\beta,\gamma) \nonumber
\end{align}
and taking the supremum over~$x$.

To show the completeness,
let $\alpha_1,\alpha_2,\ldots \in \calL$ be a Cauchy sequence under~$\dist$.
For any $x \in \Mat$,
the sequence $\alpha_1(x),\alpha_2(x),\ldots \in \barMat$ is Cauchy under~$\norm \cdot$,
and hence converges in~$\barMat$.
Define $\alpha(x) = \lim_{i \to \infty} \alpha_i(x)$.
It is readily checked that $\alpha$ is $*$-linear and hence is a member of~$\calL$.
\end{proof}

\begin{lemma}
On the subset of~$\calL$, consisting of all norm-nonincreasing maps,
the evaluation map $(\alpha,x) \mapsto \alpha(x)$ is continuous.
\end{lemma}

\begin{proof}
Let $\epsilon > 0$. 
Let $(\alpha,x)$ be in the domain of evaluation map.
Let $z \in \Mat \setminus \CC\one$ be a single-site operator 
such that~$\norm z \le 1$ and~$x-z \notin \CC\one$.
We take $\delta = \epsilon / (2 + (1+\norm x)(1+\abs{\Supp(x)}))$
and an open neighborhood of~$(\alpha,x)$ consisting of~$(\beta,y)$
such that $\dist(\alpha,\beta) < \delta$ and $\norm{x- y} < \delta$.
Then,
\begin{align}
	\norm{\alpha(x) - \beta(y)}
	&\le \norm{\alpha(x) - \beta(x)} + \norm{\beta(x) - \beta(y)} \nonumber\\
	&\le \norm{\alpha(x-z) - \beta(x-z)} + \norm{\alpha(z) - \beta(z)} + \norm{x-y}\\
	&\le \dist(\alpha, \beta)\norm{x-z}(\abs{\Supp(x)} + 1) + \dist(\alpha,\beta)\norm z + \norm{x - y}\nonumber\\
	&< \delta (1 + \norm x)(1+\abs{\Supp(x)}) + \delta + \delta = \epsilon.
	 \nonumber \qedhere
\end{align}
\end{proof}

\begin{proposition}\label{prop:trnormpreserving}
Every $*$-algebra homomorphism from~$\Mat$ into~$\barMat$, a member of~$\calL$, 
preserves trace and norm.
\end{proposition}

\begin{proof}
Let $\phi$ be a $*$-algebra homomorphism.

The trace preserving property follows from the uniqueness of the trace
and that $x \mapsto \tr(\phi(x))$ is a trace.
Alternatively, we can prove it as follows.
An arbitrary operator $x \in \Mat$ can be written as the sum of 
its hermitian part~$\sum_i a_i \pi_i$
and antihermitian part~$\ii\sum_i  b_i \pi'_i$
where $a_i,b_i \in \RR$ and $\pi_i = \pi_i^\dag = \pi_i^2$, $\pi'_i = (\pi'_i)^\dag = (\pi'_i)^2$
are projectors.
Hence, we simply show that the trace of a finitely supported projector 
is preserved under~$\phi$.
Any finite supported projector is a sum of finitely many rank~$1$ mutually orthogonal 
projectors~$\tau_j$ on its finite support, so we further reduce the proof to showing 
that~$\tr(\phi(\tau_j)) = \tr(\tau_j)$.
We can complete $\{\tau_i\}$ to a resolution of identity: $\sum_{j=1}^n \tau_j = \one$.
Note that $\tr(\tau_j) = \tr(\tau_{j+1}) = 1/n$.
Observe that there is a unitary~$u$ such that $u \tau_j u^\dag = \tau_{j+1}$
for all~$j$ where $\tau_{n+1} = \tau_1$.
Since $\phi$ is a $*$-homomorphism, $\phi(u)$ is also a unitary:
$\phi(u)^\dag \phi(u) = \phi(u^\dag u) = \phi(\one) = \one$.
Therefore, $\tr(\phi(\tau_j)) = \tr(\phi( u \tau_j u^\dag )  ) = \tr( \phi(\tau_{j+1}))$
and $\sum_j \tr(\phi(\tau_j)) = 1$ implies that $\tr(\phi(\tau_j)) = 1/n$.
This complete the proof that $\phi$ preserves the trace.

For norm,
it suffices to show that $\phi$ preserves norm of positive semidefinite operators~$y \in \Mat$
since $\norm{x}^2 = \norm{x^\dag x}$ for all~$x \in \Mat$.
Let $y = \sum_{i=1}^n a_i \pi_i$ be the eigenvalue decomposition 
with orthogonal projectors~$\pi_i$ which sum to~$\one$ and real numbers~$a_1 \ge \cdots \ge a_n \ge 0$,
so $a_1 = \norm y$.
Observe that $\phi$ maps mutually orthogonal projectors to mutually orthogonal projectors.
So, $\phi(y) = \sum_i a_i \pi'_i$ where $\pi'_i$ are mutually orthogonal projectors.
For any positive semidefinite~$\rho \in \Mat$ with $\tr(\rho) \le 1$,
we have $\tr(\rho \pi'_i) = \tr(\pi'_i \rho \pi'_i) \ge 0$ that sum to~$\tr(\rho)$.
So, $\tr(\phi(y)\rho) = \sum_i a_i \tr(\pi'_i \rho)$ is a convex combination of~$a_i$.
Therefore, $\norm{\phi(y)} = \sup_\rho \tr(\phi(y) \rho) \le a_1 = \norm{ y }$.
To show the opposite inequality,
let $\epsilon \in  (0, 1)$ and choose $\tau \in \Mat$ such that $\norm{\pi'_1 - \tau} < \epsilon$.
Then, $\norm{\tau} < 1 + \epsilon$, 
implying
$\norm{\tau^\dag \tau - \pi'_1} \le \norm{\tau^\dag}\norm{\tau - \pi'_1} + \norm{\tau^\dag - \pi'_1} \norm{\pi'_1} < (1+\epsilon)\epsilon + \epsilon < 3 \epsilon$,
implying
$\tr(\tau^\dag \tau) < 1 + 3 \epsilon$.
Put $\sigma = \tau^\dag \tau / (1+3\epsilon) \succeq 0$, so $\tr( \sigma ) < 1$.
Then, $\tr(\phi(y) \sigma) \ge a_1 - O(n \epsilon)$.
Since $\epsilon$ was arbitrary, $\norm{\phi(y)} = \sup_\rho \tr(\phi(y) \rho) \ge a_1$.
\end{proof}

\subsection{Hamiltonian time evolution}

By a {\bf bounded strength strictly local Hamiltonian}~$H$ on~$\Mat$ 
we mean a collection~$\{h_X | X \subset \ZZ^\dd, \norm{h_X} \le 1, \Supp(h_X) \subseteq X \}$ 
of hermitian elements of~$\Mat$
such that
for some $R > 0$, called {\bf range},
it holds that $h_X = 0$ whenever $\diam X > R$.
The collection may have infinitely many nonzero terms;
otherwise, one would be studying finite systems.
We write $H = \sum_X h_X$ where the sum is formal to denote the collection.
The most important role of a bounded strength strictly local Hamiltonian is
that it gives a locality-preserving $\CC$-linear $*$-map (not an algebra homomorphism)
from $\Mat$ to itself,
denoted by $[\ii H, \cdot ]: y \mapsto [\ii H, y] = \sum_X \ii [h_X, y]$
where this sum is now meaningful 
since it is always finite for any given $y \in \Mat$.
It belongs to~$\calL$ and satisfies the Leibniz rule:
\begin{align}
	[\ii H, x y] = [\ii H, x] y + x [\ii H, y].
\end{align}
A member $\beta \in \calL$ is a {\bf derivation} with spread~$\ell$
if it obeys the Leibniz rule and $\Supp(\beta(x)) \subseteq \Supp(x)^{+\ell}$
for all~$x \in \Mat$.
Every bounded strength strictly local Hamiltonian gives a derivation.

By abuse of notation, for any $t \in \RR$ 
we will write $y \mapsto y(t)$,
called time evolution by~$H$ where $H$ is implicit in notation~$y(t)$,
to mean
\begin{align}
	y(t) &= \lim_{n \to \infty} y(t)_n\label{eq:timeevolve}\\
	y(t)_n &= \sum_{k = 0}^n \frac{t^k}{k!} [\ii H, y]_k
	\qquad \text{ where }[\ii H, y]_k = \begin{cases}
	[\ii H,[\ii H, y]_{k-1}]  & (k > 0)\\
	y & (k=0)
	\end{cases} . \nonumber
\end{align}
The power series may not converge in~$\Mat$,
but we will show that it absolutely converges, at least for some nonzero~$t$,
in the norm completion~$\barMat$.
Unless otherwise specified, 
$H = \sum_X h_X$ will denote a bounded strength strictly local Hamiltonian 
with range~$R$ on~$\Mat$.

\begin{lemma}\label{lem:HeisenbergRemainder}
For any integer~$k > 0$ and any $y \in \Mat$,
\begin{align}
 \frac{\norm{ [H, y]_k } }{k!}  \le C \zeta^k \abs{\Supp(y)} \norm y
\end{align}
where $C,\zeta > 0$ depends only on~$\dd$ and~$R$.
\end{lemma}

\begin{proof}
If $y \in \CC \one$, there is nothing to prove.
Let $s$ denote a site in the support of~$y \in \Mat \setminus \CC\one$.
Introduce a graph~$G_s$ with nodes labeled by~$s$ and all nonzero terms~$h_X$ of~$H$
where an edge is present iff two operators have overlapping supports 
or an operator's support contains~$s$.
The degree of~$G_s$ is bounded because $H$ is strictly local;
an upper bound on the degree depends on~$\dd$ and~$R$.
The $k$-th nested commutator is a $k$-fold finite sum
\begin{align}
\sum_{X_1,\ldots,X_k} [h_{X_1},\cdots,[h_{X_k}, y]\cdots], \label{eq:nestedComm}
\end{align}
each term of which may not vanish only if $\{s,h_{X_1},h_{X_2},\ldots,h_{X_k}\}$
defines a connected subgraph of~$G_s$ for some $s \in \Supp(y)$.
The norm of a term here is at most $2^k \norm y$.
We say that $J$ is a multiset over a graph~$G$ 
if it is a function $J : \mathrm{Vertex}(G) \to \ZZ_{\ge 0}$,
where the image of a vertex~$v$ is called the weight or multiplicity of~$v$,
and the support of a multiset~$J$ is defined as $\Supp(J) = J^{-1}(\ZZ_{> 0})$.
The total weight of~$J$ is $\abs J = \sum_{g \in G} J(g)$.
We say that a multiset is connected if the support is connected.
Given a tuple~$(s,X_1,\ldots,X_k)$ we get a unique multiset~$J$ over~$G_s$ of total weight~$k+1$,
whose support is connected in~$G_s$.
Conversely, given a connected multiset~$J$ such that $J(s) =1$ and $\abs J = k+1$,
there are at most~$k!$ different tuples that correspond to~$J$.
Therefore,
\begin{align}
 \frac{1}{k!} \norm{ [H, y]_k } \le \sum_{s \in \Supp(y)} \quad \sum_{\mathrm{multiset}\,J:\, y \in \Supp(J) \subseteq G_s,\, \text{connected},\, \abs J = k+1 } 2^k \norm y .
\end{align}
Hence, it remains to show that the number of rooted, connected multisets over~$G_s$ with total weight $k+1$
is at most exponential in~$k$.
When the weight for each node is~$1$ or~$0$,
this is a subgraph counting and it is well known~\cite[Lemma~2.1]{Borgs2010} that
an upper bound is~$(ed)^{k+1}$ 
where $e \approx 2.718$ and $d$ is the maximum degree of~$G_s$.
Multiplicities increase the count,
but a bound is still exponential in~$k$: 
$\sum_{j=1}^{k+1} \binom{k}{j-1} (ed)^{j} \le (1+ed)^{k+1}$.
\end{proof}

Using~\ref{lem:HeisenbergRemainder}, we set
\begin{align}
t_0 = \frac 1 {128 \zeta}.
\end{align}
By~\ref{lem:HeisenbergRemainder}, for $t \in (-t_0,t_0)$ 
the sequence~$y(t)_0, y(t)_1, y(t)_2,\ldots$
is a Cauchy sequence, and hence converges in~$\barMat$ absolutely.
Therefore, the short time evolution belongs to~$\calL$ and preserves~$\one$.
The following inequalities will be handy
to show that short time evolution is a $*$-automorphism of~$\barMat$.

\begin{lemma}\label{lem:someIneq}
There exist constants $C',C'' > 0$ such that
for any $t,t' \in (-t_0, t_0)$, any $x,y\in \Mat$,
and any positive integers $n,m$ such that $n \le m$,
it holds that
\begin{align}
\norm{ y(t)_n(t')_m - y(t+t')_n } &\le 2^{-n} C' \norm y \abs{\Supp(y)} ,\\
\norm{ x(t)_n y(t)_n - (xy)(t)_n } & \le 2^{-n} C'' \norm x \norm y \abs{\Supp(x)} \abs{\Supp(y)}. \nonumber
\end{align}
\end{lemma}

\begin{proof}
This is a straightforward calculation using~\ref{lem:HeisenbergRemainder} and using the Leibniz rule,
which implies
\begin{equation}
\sum_{j=0}^n \binom{n}{j} [\ii H, x]_j [\ii H, y]_{n-j} = [\ii H, xy]_n. \qedhere
\end{equation}
\end{proof}

\begin{corollary}\label{lem:TEisHomo}
Any short time evolution by~$H$ for time~$t$ with $\abs t < t_0$
is a $*$-algebra homomorphism from~$\Mat$ to~$\barMat$
preserving trace and norm.
\end{corollary}

\begin{proof}
Use the second inequality of~\ref{lem:someIneq}.
Apply~\ref{prop:trnormpreserving}.
\end{proof}

Being continuous, the map~$y \mapsto y(t)$ can be extended to the entire~$\barMat$
uniquely:
if $\bar y = \lim_{n \to \infty} y_n \in \barMat$ where $y_n \in \Mat$,
then we define
\begin{align}
\bar y(t) = \lim_{n \to \infty} y_n(t).
\end{align}
This limit exists by the metric completeness
because $\{y_n(t)\}_n$ is Cauchy whenever $\{y_n\}$ is, implied by~\ref{lem:TEisHomo}.
It is obvious that $\bar y(t)$ does not depend on the sequence~$y_n$ converging to~$\bar y$.

\begin{corollary}\label{thm:TEisAutoNormPreserving}
For any bounded strength strictly local Hamiltonian~$H$,
the time evolution $U^H_t : x \mapsto x(t)$ for any $t \in (-t_0,t_0)$
gives a $*$-automorphism of~$\barMat$ preserving trace and norm.
\end{corollary}

\begin{proof}
An inverse exists: $x \mapsto x(t) \mapsto x(t)(-t)$
where $x(t)(-t) = x$ by~\ref{lem:someIneq}.
Use~\ref{lem:TEisHomo}.
\end{proof}

Although we have not shown that $U^H_t$ is well defined for arbitrary~$t\in \RR$ 
by the formula~\eqref{eq:timeevolve},
since $U^H_{t'}$ for $\abs{t'} < t_0$ is now an automorphism on~$\barMat$,
we can define $U^H_t$ for arbitrary~$t$ by a composition of finitely many~$(U^H_{t_0/2})^{\pm 1}$ 
and some $U^H_{t'}$ for $\abs{t'} < t_0$.
It follows from~\ref{lem:someIneq} that $U^H_t U^H_{t'} = U^H_{t'} U^H_t$ for $t,t' \in (-t_0, t_0)$.
Hence, the composition can have arbitrary order,
and we have a group homomorphism from the additive group of~$\RR$
into the group of all $*$-automorphisms of~$\barMat$,
generated by~$H$.%
\footnote{
This statement is proved in~\cite[Chap.~6]{BratteliRobinson2}.
The approach is as follows.
First, one analyzes finite volume time evolution~$e^{\ii t H_\Omega}$
for any finite set~$\Omega \subset \ZZ^\dd$, 
evaluated at a fixed operator~$x \in \Mat$.
Then, using a Lieb--Robinson bound~\cite{Hastings2010}
one shows that the infinite volume limit exists.
This approach works for all~$t \in \RR$ at once.
We contend ourselves with a simple proof.
}

\begin{lemma}\label{thm:timecontinuous}
The path~$t \mapsto U^H_t$ is continuous.
\end{lemma}

\begin{proof}
It suffices to show the continuity at~$t=0$.
Let $\epsilon > 0$.
For any $x \in \Mat$, if~\mbox{$|t| < \min(\epsilon / 4 C \zeta, t_0)$} 
where the constants~$C,\zeta$ are from~\ref{lem:HeisenbergRemainder},
then
\begin{align}
	\norm{U^H_t(x) - x} 
	&= 
	\norm*{\sum_{k \ge 1} [\ii t H, x]_k / k!} 
	\le	
	 C \frac{\zeta t}{ 1- \zeta \abs t } \norm x \abs{\Supp(x)}, \\
	\dist(U^H_t, U^H_0) &= \sup_{x \in \Mat\setminus \CC\one} \frac{\norm{U^H_t(x) - U^H_0(x)}}{\norm x \abs{\Supp(x)}} 
	<
	\epsilon . \nonumber \qedhere
\end{align}
\end{proof}

\begin{remark}\label{rem:continuity}
The seemingly trivial statement of~\ref{thm:timecontinuous}
would have been false if we used a different topology on~$\calL$.
This will give some reason why we wanted the factor~$\abs{\Supp(x)}$ in the denominator
when we introduced~$\dist$.
Let $H = \half \sum_j \sigma_j^Z$ be a Hamiltonian on $\Mat(\ZZ,p=2)$.
This is a noninteracting Hamiltonian on a one-dimensional chain of qubits;
being one-dimensional will be immaterial.
For any positive integer~$n$, let 
\begin{align}
g_n(\phi) = \half (\ket{0^{\otimes n}} + \phi \ket{1^{\otimes n}})(\bra{0^{\otimes n}} + \phi^{-1} \bra{1^{\otimes n}})
\end{align}
be a projector where $\phi$ is a complex phase factor of magnitude one.
Here, $\sigma^Z \ket 0 = \ket 0$ and $\sigma^Z \ket 1 = - \ket 1$.
The precise support of $g_n(\phi)$ does not matter.
Calculation shows that $U_t^H(g_n(1)) = g_n(e^{\ii t n})$ and thus
\begin{align}
\sup_{x \neq 0} \frac{ \norm{U^H_t(x) - U^H_0(x)}}{\norm x} 
\ge \sup_n \norm{g_n(e^{\ii t n}) - g_n(1)} = \begin{cases} 
1 & \text{ if $t$ is an irrational multiple of $\pi$ }\\
0 & \text{ if } t = 0
\end{cases}.
\end{align}
Therefore, in the induced norm topology of~$\calL$
the time evolution is not continuous in time.
\hfill $\square$
\end{remark}

\begin{proposition}\label{thm:TimeEvolutionIsDifferentiable}
The time evolution by any bounded strength strictly local Hamiltonian 
is differentiable with respect to time:
\begin{align}
	\lim_{t \to 0} \dist\left(\frac 1 t (U^H_{s+t} - U^H_s) , ~~U^H_s \circ [\ii H, \cdot] \right) = 0.
\end{align}
\end{proposition}

\begin{proof}
Since $U^H_s$ preserves $\dist$,
it suffices to consider $s = 0$.
From~\ref{lem:HeisenbergRemainder}, we see for $\abs t < t_0$
and any~$x\in \Mat$,
\begin{align}
\frac 1 t (U^H_t(x) - x) - [\ii H, x] 
&=
\sum_{k = 2}^{\infty} \frac{t^{k-1}}{k!} [\ii H, x]_k , \nonumber \\
\norm*{\frac 1 t (U^H_t(x) - x) - [\ii H, x] } 
&\le 
\sum_{k \ge 2} \abs t ^{k-1} C \zeta^k \norm{x} \abs{\Supp(x)} = \frac{ C \zeta^2 \abs{t}} {1 - \zeta \abs t} \norm x \abs{\Supp(x)},\label{eq:diffRemainder}\\
\dist\left( \frac{U^H_t - U^H_0}{t - 0}, [\ii H, \cdot] \right) &\le \frac{C\zeta^2\abs t}{1-\zeta \abs t}
\xrightarrow{\quad t \to 0 \quad } 0. \qedhere
\end{align}
\end{proof}

\begin{remark}\label{rem:reasonSupp}
With the differentiability, 
we can give a stronger reason that we want the factor~$\abs{\Supp(x)}$
in the definition of~$\dist$.
Consider the same Hamiltonian~$H = \sum_j \sigma^Z_j$ as in~\ref{rem:continuity}.
Let $x_n = \bigotimes_{j=1}^n \sigma^X_j$ be a tensor product of Pauli matrices,
which has $\norm{x_n} = 1$ for all~$n \ge 1$.
The derivation $x \mapsto [\ii H,x]$ by~$H$
blows the norm up:
\begin{align}
[\ii H, x_n] &= - 2 \sum_{j=1}^n \sigma^Y_j \prod_{i \neq j} \sigma^X_i,\nonumber\\
\norm{ [\ii H, x_n] } &= 2 n = 2 \abs{\Supp(x_n)}, \\
\norm*{\frac 1 t (U^H_t(x_n) - x_n) - [\ii H, x_n] } &\ge 2 \abs{\Supp(x_n)} - 2/t \qquad (t > 0),
\nonumber
\end{align}
where the norm calculation uses the fact that different summands in the first line commute.
The third line shows that
in order for~$U^H_t$ to be differentiable in~$t$ with respect to a
metric
\begin{align}
\dist_\eta(\alpha,\beta) = \sup_{x \in \Mat \setminus \CC\one} \frac {\norm{\alpha(x) - \beta(x)}}{\norm x \abs{\Supp(x)}^\eta},
\end{align}
we must have $\eta \ge 1$.
Our metric~$\dist$ is~$\dist_{\eta =1}$.

If we consider a sequence $x_n / \sqrt n$ that converges to zero as~$n \to \infty$ 
in~$\Mat$ with the norm topology,
we see that $\lim_{n \to \infty} \norm{[\ii H, x_n / \sqrt n]} = \infty$.
This shows that the derivation by~$H$ is \emph{not} norm-continuous.
This may be thought of as the origin that we need the factor~$\abs{\Supp(x)}$
in~$\dist$.~$\square$
\end{remark}

\begin{remark}\label{rem:strongTopology}
Our metric topology is strictly finer than the strong topology.
To show this, we find a sequence in~$\calL$ which converges in the strong topology, 
but not in our $\dist$ topology.
Consider a depth~1 quantum circuit~$P$
consisting of single-site unitaries~$\sigma^X_j$ (a Pauli operator)
on every site~$j$ of~$\Mat(\ZZ,p=2)$.
For each integer~$n > 0$, 
we define $P_n = \prod_{j=-n}^n \sigma^X_j$, a unitary of~$\Mat(\ZZ,2)$.
For any $x \in \Mat(\ZZ,2)$ we know $x$ is supported on~$[-n_x, n_x] \subset \ZZ$
for some $n_x > 0$.
Hence, $P_n(x) \equiv P_n x P_n^\dag = P(x)$ if $n > n_x$.
In particular, for all~$x \in \Mat$
\begin{align}
\lim_{n \to \infty} \norm{P_n(x) - P(x)} = 0.
\end{align}
That is, the sequence~$\{ P_n \in \calL \}$ converges to~$P \in \calL$ in the strong topology.
However, $P_n(\sigma^Z_{k}) = \sigma^Z_{k}$ if $k > n$
whereas $P(\sigma^Z_{k}) = - \sigma^Z_{k}$.
Therefore,
\begin{align}
\dist(P_n, P) = \sup_{x \in \Mat \setminus \CC\one} \frac{\norm{ P_n(x) - P(x)}}{\norm x \abs{\Supp(x)}} \ge \sup_k \norm{P_n(\sigma^Z_k) - P(\sigma^Z_k)} = 2,
\end{align}
which does not converge to zero as $n \to \infty$.
\hfill $\square$
\end{remark}

\subsection{A limit of quantum circuits}

Given any region $\Omega \subseteq \ZZ^\dd$, we define 
\begin{align}
	H_\Omega = \sum_{X \subseteq \Omega} h_X,
\end{align}
the collection of all terms of~$H$ supported on~$\Omega$.

\begin{lemma}[A Lieb--Robinson bound~\cite{HHKL2018}]\label{lem:LR}
For any $t \in (-t_0,t_0)$ and any $\Omega \subseteq \ZZ^\dd$,
\begin{align}
\norm{
	U^{H}_t(x) - U^{H_\Omega}_t(x)
}
\le 
C (\abs t / 4 t_0)^{L/R} \norm x \abs{\Supp(x)}
\end{align}
where $L$ is the distance between~$\Supp(x)$ and the complement of~$\Omega$,
and the constant $C$ depends only on~$R,\dd$.
\end{lemma}

\begin{proof}
The defining expansions~\eqref{eq:timeevolve} 
are the same up to $k$-th order, where $k = L/R$,
beyond which the series converges geometrically by~\ref{lem:HeisenbergRemainder}.
\end{proof}

\begin{lemma}[Lemma~6 of~\cite{HHKL2018}]\label{lem:Staircase}
Let $A, B, C \subset \ZZ^\dd$ be pairwise disjoint finite subsets.
Then,
\begin{align}
\norm*{
e^{\ii t H_{A\cup B \cup C}}
-
e^{\ii t H_{A \cup B}} e^{-\ii t H_B} e^{\ii t H_{B \cup C}}
}\le
C \abs t e^{-L/R} \abs{ A^{+R} }
\end{align}
for some $C > 0$ depending only on~$\dd,R$.
Here, $L$ is the $\ell_\infty$-distance (that is our convention for~$\ZZ^\dd$) 
between~$A^{+R}$ and~$C$.
\end{lemma}

This is slightly weaker than~\cite[Lemma~6]{HHKL2018}
as the upper bound depends on the volume of~$A$.
We will use this lemma with small~$A$.

\begin{proof}
Since $A,B,C$ are finite, those exponentials are unitaries of~$\Mat$.
Without loss of generality, assume $t > 0$.
Suppress the union symbol~$\cup$.
Put $W(s) = e^{\ii s H_{A B  C}} e^{-\ii s H_{B C}}$ 
and $V(s) = e^{\ii s H_{A B}} e^{-\ii s H_{B}}$ for~$s \in [0,t]$.
They are unique solutions to first order differential equations 
$\partial_s W(s) = \ii W'(s) W(s)$ and $\partial_s V(s) = \ii V'(s) V(s)$
with initial condition $W(0) = V(s) = \one$,
where
\begin{align}
\ii W'(s) &= 
U^{H_{A B C}}_s (H_{A B C} - H_{B C}) = U^{H_{A B C}}_s (H_{A^{+R}}),\nonumber \\
\ii V'(s) &= 
U^{H_{A B}}_s (H_{A B} - H_{B}) = U^{H_{A B}}_s (H_{A^{+R}}),\\
\norm{W'(s) - V'(s)} &\le \sum_{X \subseteq A^{+R}} \norm{U_s^{H_{ABC}}(h_X) - U_s^{H_{AB}}(h_X) } 
\le O(1) \abs*{A^{+R}} R^\dd e^{-L/R} \nonumber & \text{by \ref{lem:LR}}.
\end{align}

On the other hand, a differential 
equation~$\partial_s Y(s) = \ii V(s)^\dag (-V'(s) + W'(s) ) V(s) Y(s)$ with initial condition~$Y(0) = \one$,
gives a unique unitary solution~$Y(s) = V(s)^\dag W(s)$ for~$s \in [0,t]$.
Jensen's inequality~$\norm{Y(t) - \one} \le \int_0^t \rd s \norm{\partial_s Y(s)}$
gives the result:
\begin{align}
\norm*{
e^{\ii t H_{A B C}}
-
e^{\ii t H_{A B}} e^{-\ii t H_B} e^{\ii t H_{B C}}
} 
&=
\norm*{
e^{\ii t H_{A B C}} e^{-\ii t H_{B  C}}
-
e^{\ii t H_{A B}} e^{-\ii t H_B} 
}\\
& = \norm*{ W(t) - V(t) } = \norm*{Y(t) - \one} \le O(1) t e^{-L/R} \abs{A^{+R}}. 
\nonumber \qedhere
\end{align}
\end{proof}

\begin{lemma}\label{lem:ShortTimeEvolutionIsLimit}
Given any $L > 10R$ and $t \in \RR$ with $\abs t < t_0$,
there exists a quantum circuit~$Q_L$ of spread~$10(2\dd+1)L$
such that
\begin{align}
\dist(U^H_t, Q_L) \le C \abs t L^{2\dd} e^{-L/R}
\end{align}
for some constant~$C$ depending on~$\dd,R$.
\end{lemma}

\begin{proof}
Assume $t > 0$ for brevity.

We use the approximating circuit constructed in~\cite{HHKL2018}.
It is built out of local unitaries,
each of which is supported on a ball of diameter at most~$L$.
Divide~$\ZZ^\dd$ into disjoint ``cells'' 
each of which contains a radius~$7L$ ball and is contained in the concentric radius~$10L$ ball 
such that the cells are colored with $\dd+1$ colors $\{0,1,\ldots,\dd\}$,
and any two cells of a given color are separated by distance~$5L$.
Such a division exists as seen for example by considering a regular triangulation of~$\RR^\dd$,
and fattening all cells. For $\dd=2$, one can consider a honeycomb tiling.
Let $ce(k)$ denote the union of all cells colored~$k = 0,1,2,\ldots,\dd$,
and let $ce(k)^+$ denote the $L$-neighborhood of~$ce(k)$.
The $(2\dd+1)$st layer, the last, 
of the circuit is $e^{\ii t H_{ce(\dd)^+}}$.
Here, we abused the notation; the Hamiltonian~$H_{ce(k)^+}$ 
consists of terms that act on disjoint cells of linear size $\sim L$,
so the exponential is an infinite product unitaries,
which must be interpreted merely as a collection of those local unitaries.
The layer beneath, the $2\dd$-th, is $e^{-\ii t H_{ce(\dd)^+ \cap (\ZZ^\dd \setminus ce(\dd))}}$.
Inductively, having defined $2k$-th layer for $k > 1$,
we define two more layers 
by declaring that the starting Hamiltonian is~$H_{\bigcup_{j=0}^{k-1} ce(j)}$.
This is a Hamiltonian on $k$-colored cells,
and gives circuit elements down to the second layer.
The first layer is simply $U_t^{H_{ce(0)}}$.
This defines~$Q_L$ of depth~$2\dd+1$ and spread~$10(2\dd+1)L$.

Next, we have to show that this is a desired quantum circuit.
Let $x \in \Mat \setminus \CC\one$.
Let $\Omega \subset \ZZ^\dd$ be a finite subset consisting of the cells 
that intersect~$\Supp(x)^{+123(2\dd+1)L}$.
So, $\abs{\Omega} \le O(1) L^\dd \abs{\Supp(x)}$.
By~\ref{lem:LR}, we have 
\begin{align}
\norm{U^H_t(x) - U^{H_\Omega}_t(x)} \le O(1) e^{-L/R} \norm x \abs{\Supp(x)}.
\end{align}
Apply the circuit construction above to~$H_\Omega$, keeping the cells intact.
The gates in this $\Omega$-dependent circuit~$Q'$ are exactly the same as those of~$Q_L$,
except those near the boundary of~$\Omega$.
But $Q'$ contains the ``geometric lightcone'' of~$\Supp(x)$,
and therefore $Q'(x) = Q_L(x)$.

Now, by~\ref{lem:Staircase} we can estimate~$\norm{e^{\ii t H_\Omega} - Q'}$.
Consider one cell associated with the top layer, the $(2\dd+1)$st.
Call this cell~$A$ and put $B = A^{+L} \setminus A$ and $C = \Omega \setminus (A\cup B)$.
The distance between $A^{+R}$ and $C$ is at least~$L$.
The gate~$e^{\ii t H_{A\cup B}}$ is precisely a gate of~$Q'$ at $(2\dd+1)$st layer,
and $e^{-\ii t H_B}$ is a gate at the $2\dd$th layer of~$Q'$.
The full two top layers of~$Q'$ are obtained by such decompositions, 
iterated at most $\abs \Omega$ times, incurring accumulated ``error'' 
$O(1) t e^{-L/R} L^\dd \abs{\Omega}$.
Inductively, similar decompositions give the full $Q'$.
Since $\abs{\Omega} \le \abs{\Supp(x)}(2 L)^\dd$,
we have 
\begin{align}
\norm{e^{\ii t H_\Omega} - Q'} \le
O(1) \dd t e^{-L/R} L^\dd \abs{\Omega} \le O(1) t e^{-L/R} L^{2\dd} \abs{\Supp(x)}.
\end{align}
We conclude that
\begin{align}
\norm{U^H_t(x) - Q_L(x)} 
&= \norm{U^H_t(x) - Q'(x)} \nonumber\\
&\le 
\norm{U^H_t(x) - U^{H_\Omega}_t(x)} 
+
2 \norm{e^{\ii t H_\Omega} - Q'} \norm x \\
&\le O(1) t e^{-L/R} L^{2\dd} \norm x \abs{\Supp(x)} \nonumber
\end{align}
where we abused notations to write~$Q'$ for the unitary 
and~$Q'(x)$ for the conjugation of $x$ by the unitary.
\end{proof}

\begin{theorem}\label{thm:TimeEvolutionIsLimit}
The time evolution~$U^H_t$ by any bounded strength strictly local Hamiltonian~$H$ on~$\Mat(\ZZ^\dd)$ 
for any time $t \in \RR$
is a limit of a sequence of finite depth quantum circuits.
Therefore, the $\dist$-closure of the group of all QCA 
contains the time evolution by any bounded strength strictly local Hamiltonian.
\end{theorem}

\begin{proof}
Take a positive integer $n$ sufficiently large that $\tau  = t/ n \in (-t_0, t_0)$.
Let $Q_L$ be a quantum circuit corresponds to~$U^H_\tau$ by~\ref{lem:ShortTimeEvolutionIsLimit}.
Let $x \in \Mat\setminus \CC\one$.
Then, for any $k = 1,2,\ldots, n$,
\begin{align}
&\norm*{
\left(U^H_\tau \right)^k \left(Q_L\right)^{n-k}(x) 
-
\left(U^H_\tau \right)^{k-1} \left(Q_L\right)^{n-k+1}(x) 
}\nonumber\\
&=
\norm*{
U^H_\tau(Q_L^{n-k}(x) ) - Q_L( Q_L^{n-k}(x) )
} & \text{by~\ref{thm:TEisAutoNormPreserving}}\\
&\le \dist(U^H_\tau, Q_L) \norm x \abs{\Supp(Q_L^{n-k}(x))} & \text{by definition~\ref{def:dist}}\nonumber\\
&\le C t_0 L^{2\dd} e^{-L/R} \norm x (20(2\dd+1)nL)^\dd \abs{\Supp(x)} & \text{by~\ref{lem:ShortTimeEvolutionIsLimit}}.\nonumber
\end{align}
Therefore,
\begin{align}
\dist(U^H_t, Q_L^{n}) \le 
\sum_{k=1}^n \dist(
\left(U^H_\tau \right)^{k} \left(Q_L\right)^{n-k},
\left(U^H_\tau \right)^{k-1} \left(Q_L\right)^{n-k+1} 
)
\le C' n^{\dd+1} L^{3\dd} e^{-L/R}
\end{align}
for some constant $C'$ that depends only on~$\dd,R$.
Since $\calL$ is $\dist$-complete by~\ref{thm:distIsMetric},
taking $L \to \infty$ we complete the proof.
\end{proof}

\cite[Thm.~5.6]{Ranard2022} says that a certain class of 
approximately locality-preserving $*$-automorphisms~$\alpha$ of~$\barMat(\ZZ^\dd)$ with~$\dd=1$,
is contained in the $\dist$-completion of the group of QCA.
To state the studied class of~$\alpha$,
we denote by~$[S] \subset \ZZ$ the smallest interval containing~$S$ for any $S \subset \ZZ$.
The required condition for~$\alpha$ is that
for any given $\epsilon > 0$ 
there is a length scale~$\ell$ (corresponding to the spread of a QCA)
such that for every $x \in \Mat$ there exists $y_x \in \Mat([\Supp(x)]^{+\ell})$
with $\norm{\alpha(x) - y_x} < \epsilon \norm{x}$.
A specific relation between $\epsilon$ and $\ell$ encodes the decay rate of the tails of~$\alpha$.
It appears that an important condition is that 
$\ell$ should be uniform across the lattice,
but not the specific decay rate.
In course of proving~\cite[5.6]{Ranard2022},
they also show~\cite[4.10]{Ranard2022} 
that two one-dimensional QCA on the same~$\Mat(\ZZ,p)$ 
that are $\epsilon_0$-close in~$\dist$
must blend where $\epsilon_0$ is a positive threshold depending only on the spread.
This raises a natural question:
in higher dimensions, if two QCA are $\epsilon$-close in~$\dist$,
do they give boundary algebras of the same Brauer class?

\section{One-dimensional Brauer group is trivial}\label{app:1d}

In this appendix, 
we are going to show that the Brauer group of invertible subalgebras in one dimension is trivial.
The Brauer triviality of an invertible subalgebra means 
that there is a locality-preserving isomorphism into the invertible subalgebra
(up to stabilization) from a full local operator algebra.
In fact, we will not need stabilization
and will find a decomposition of the invertible subalgebra
into mutually commuting central simple $*$-subalgebras on intervals of~$\ZZ$.
These interval subalgebras are images of single-site matrix algebras.
It is not too difficult to find some central subalgebra around an interval.
Then, we are led to think about central subalgebras on far separated intervals,
and try to fill the ``gaps'' by taking commutants.
We will do this, but there is one difficulty.
Namely, it does not immediately follow
that the commutants are again central
since the existence argument~(\ref{lem:CentralClosure0}) 
of a central algebra on an interval 
is too abstract to give any useful information 
about potentially central elements.
So, we take a ``dual'' construction for interval central subalgebras,
defined as commutants of some local operators
such that any potential central elements have restricted support.

For the rest of this section,
we let $\calA \subseteq \Mat(\ZZ)$ denote any one-dimensional 
invertible subalgebra of spread~$\ell$.
For any interval~$[a,b] \cap \ZZ$, we introduce an abbreviation
\begin{align}
	\calA_{[a,b]} = \calA \cap \Mat( [a,b] \cap \ZZ ).
\end{align}
We also use a notation for any two sets~$\calX, \calZ \subseteq \calA$
\begin{align}
\Comm( \calX \,|\, \calZ ) = \{ z \in \calZ ~|~ \forall x \in \calX :~[z,x] = 0 \}.
\end{align}
In this notation, $\calX$ need not be a subset of~$\calZ$.
The center of an algebra~$\calY$ will be denoted by
\begin{align}
\Cent( \calY ) = \Comm(\calY\,|\, \calY).
\end{align}

First we will need a tool to examine operator components individually 
that are far apart in space.

\begin{lemma}\label{lem:LocallyFactorizable}
Suppose an element~$x$ of an invertible subalgebra~$\calA \subseteq \Mat(\ZZ^\dd)$
of spread~$\ell$
is decomposed as $x = \sum_i y_i z_i$ 
such that $\{y_i\}$ and $\{z_i\}$ are each linearly independent
(obtained by e.g.~the Schmidt decomposition with respect to the Hilbert--Schmidt inner product).
If $Y = \bigcup_i \Supp(y_i)$ and~$Z = \bigcup_i \Supp(z_i)$
are separated by distance $> \ell$,
then $y_i , z_i \in \calA$ for all~$i$.
\end{lemma}

In~\cite{FreedmanHastings2019QCA}
such a subalgebra~$\calA$ is called ``locally factorizable.''

\begin{proof}
The commutant~$\calB$ of~$\calA$ within~$\Mat(\ZZ^\dd)$
is $\ell$-locally generated (\ref{lem:csa}, \ref{lem:InvertibilityIsSymmetric}).
Let $b$ be any $\ell$-local generator of~$\calB$ whose support overlaps with~$Z$.
Since~$Y$ and~$Z$ are far apart, $b$ does not overlap with any~$y_i$.
Hence, $0 = [b,x] = \sum_i y_i [b,z_i] =  \sum_i y_i \otimes [b,z_i]$.
The linear independence of~$\{y_i\}$ implies 
that $[b,z_i] = 0$ for all~$i$.
This means that each $z_i$ commutes with every generator of~$\calB$.
Therefore, $z_i$ belongs to the commutant of~$\calB$, 
which is $\calA$ by~\ref{lem:InvertibilityIsSymmetric}.
A symmetric argument shows that~$y_i \in \calA$ for all~$i$.
\end{proof}

The following~\ref{lem:CentralClosure0} 
is a seed for central subalgebras of an invertible subalgebra,
not necessarily on one dimension.
Indeed, this lemma has nothing to do with locality.
When we use it, the assumption of this lemma 
will be fulfilled always by~\ref{lem:CommutantOfRestricted}.

\begin{lemma}[Lemma~3.8 of~\cite{FreedmanHastings2019QCA}]\label{lem:CentralClosure0}
Let $\calX \subseteq \calZ$ be unital $*$-subalgebras of  
a full matrix algebra of dimension~$\Tr(\one) < \infty$.
Let~$\{ \pi_\mu  \}$ and~$\{ \tau_\nu \}$ be the central minimal projectors
of~$\calX$ and~$\calZ$, respectively.
Suppose 
\begin{align}
	\frac{\Tr(\pi_\mu \tau_\nu)}{\Tr(\one)} 
	= \frac{\Tr(\pi_\mu)}{\Tr(\one)} \frac{\Tr(\tau_\nu)}{\Tr(\one)}
\end{align}
for all~$\mu,\nu$.
Then, there exists a central (and hence simple) $*$-subalgebra~$\calY$ such that
$\calX \subseteq \calY \subseteq \calZ$.
\end{lemma}

The proof below is the same as that in~\cite{FreedmanHastings2019QCA}.

\begin{proof}
For the simple subalgebras~$\pi_\mu \calX$ and~$\tau_\nu \calZ$,
let~$[\pi_\mu \calX]$ and~$[\tau_\nu \calZ]$ 
denote some full matrix algebras isomorphic to them
so that we have $\Tr$-preserving isomorphisms for all $\mu, \nu$
\begin{align}
\pi_\mu \calX &\cong [\pi_\mu \calX] \otimes \one_{c^\calX_\mu}, &
\tau_\nu \calZ &\cong [\tau_\nu \calZ] \otimes \one_{c^\calZ_\nu} ,\\
\Tr(\pi_\mu) &= c^\calX_\mu \dim [\pi_\mu \calX], &
\Tr(\tau_\nu) &= c^\calZ_\nu \dim [\tau_\nu \calZ].\nonumber
\end{align}
Given a pair~$(\nu,\mu)$, consider a homomorphism
\begin{align}
\varphi_{\nu,\mu} : 
[\pi_\mu \calX] \xrightarrow{\quad\otimes \one_{c^\calX_\mu}\quad} \pi_\mu \calX \xrightarrow{\quad \tau_\nu \times \quad } \tau_\nu \calZ \xrightarrow{\qquad} [\tau_\nu \calZ].
\end{align}
Since $[\pi_\mu \calX]$ is simple,
this map~$\varphi_{\nu,\mu}$ can be either zero or an injection.
Hence, there is a nonnegative integer~$C_{\nu,\mu}$
that counts the multiplicity of~$[\pi_\mu \calX]$ in the image.
That is, $\varphi_{\nu,\mu}([\pi_\mu \calX]) \cong [\pi_\mu \calX] \otimes \one_{C_{\nu,\mu}}$
where $\cong$ is a $\Tr$-preserving isomorphism.
Tracking the traces of the unit,
we have
\begin{align}
C_{\nu, \mu} \frac{\Tr(\pi_\mu)}{c^\calX_\mu} 
= \frac{\Tr(\tau_\nu \pi_\mu)}{c^\calZ_\nu}.
\end{align} 
By assumption, this is equal to $\Tr(\tau_\nu) \Tr(\pi_\mu)/\Tr(\one) c^\calZ_\nu$.
Rearranging, we have
\begin{align}
C_{\nu,\mu} 
= \frac{\Tr(\tau_\nu)}{c^\calZ_\nu} \frac{c^\calX_\mu }{\Tr(\one)} 
= \frac{\dim [\tau_\nu \calZ] \gcd}{\Tr(\one)} \frac{c^\calX_\mu}{\gcd}
\end{align}
where $\gcd = \gcd(\{c^\calX_\mu\}_\mu)$.
If $n_\mu$ are integers such that $\sum_\mu n_\mu c^\calX_\mu = \gcd$,
we see that 
\begin{align}
N_\nu = \frac{(\dim [\tau_\nu \calZ]) \gcd}{\Tr(\one)} = \sum_\mu C_{\nu,\mu} n_\mu
\end{align} 
is a positive integer for each~$\nu$.
It follows that
\begin{align}
\bigoplus_\mu \varphi_{\nu,\mu}([\pi_\mu \calX]) 
\,\cong\,
\bigoplus_\mu [\pi_\mu \calX] \otimes \one_{C_{\nu,\mu}} 
\,\cong\,
\bigoplus_\mu [\pi_\mu \calX] \otimes \one_{c^\calX_{\mu}/\gcd}
 \otimes \one_{N_\nu}
\subseteq [\tau_\nu \calZ]
\end{align}
where all the isomorphisms are $\Tr$-preserving.
Since~$\calX$ is a subalgebra of~$\calZ$,
the $c^\calZ_\nu$-weighted direct sum over~$\nu$ of these images must be equal to~$\calX$ itself:
\begin{equation}
\calX \cong
\bigoplus_{\nu} \Big( \bigoplus_\mu [\pi_\mu \calX] \otimes \one_{c^\calX_{\mu}/\gcd} \Big) \otimes \one_{N_\nu c^\calZ_\nu} .
\end{equation}
Now, let $d = \sum_\mu (c^\calX_\mu / \gcd) \dim [\pi_\mu \calX]$.
Observe that $d \cdot N_\nu c^\calZ_\nu = \sum_\mu (c^\calX_\mu / \gcd) \dim [\pi_\mu \calX] N_\nu c^\calZ_\nu = c^\calZ_\nu \sum_\mu C_{\nu,\mu} \dim [\pi_\mu \calX]$.
The last sum~$\sum_\mu C_{\nu,\mu} \dim [\pi_\mu \calX]$ 
is equal to~$\dim [\tau_\nu \calZ]$ because $\calX$ is a unital subalgebra of~$\calZ$.
Therefore, $d \cdot N_\nu c^\calZ_\nu = \Tr(\tau_\nu)$.
This dimension counting shows that a desired subalgebra~$\calY$ can be defined as
\begin{equation}
\calY \cong \left\{ 
\bigoplus_{\nu} M \otimes \one_{N_\nu c^\calZ_\nu}~\middle|~ M \text{ is any $d \times d$ matrix } \right\} \qedhere
\end{equation}
\end{proof}

The following lemma is another tool that we will find useful later.
Namely, the commutant of certain local operators is locally generated.

\begin{lemma}\label{lem:localGenerators}
Let $\calD^\mathrm{left} \subseteq \calA_{[a-10\ell, a+10\ell]}$ 
and $\calD^\mathrm{right} \subseteq \calA_{[b-10\ell, b+10\ell]}$ 
be $*$-subalgebras.
For $b - a > 40\ell$, a $*$-subalgebra
\begin{align}
\calC &= \Comm(\calD^\mathrm{left} \calD^\mathrm{right} \,|\,  \calA_{[a, b]} ) 
\end{align}
is $20\ell$-locally generated.
\end{lemma}

The subsets $\calD^\mathrm{left,right}$ can be~$\CC$,
in which case $\calC = \calA_{[a,b]}$.

\begin{proof}
$\calC$ contains $\calA_{[a + 12\ell, b - 12\ell]}$ and $\calA_{[a+16\ell, b- 16\ell]}$.
By~\ref{lem:CommutantOfRestricted},
the centers of these two algebras have disjoint support,
and \ref{lem:CentralClosure0} gives a central subalgebra~$\calC_0$ 
sandwiched by these two subalgebras.
In the finite dimensional $*$-algebra~$\calC$,
the central subalgebra~$\calC_0$ must be a tensor factor:
$\calC$ is unitarily isomorphic 
to~$\bigoplus_\mu \pi_\mu \calC \subseteq \Mat([a-10\ell,b+10\ell])$
for some minimal central projectors~$\pi_\mu \in \calC$,
and each simple summand~$\pi_\mu \calC$ must have $\calC_0$ as a tensor factor.
So, we have $\calC \cong \calC_0 \otimes \calM$ where $\calM = \Comm(\calC_0 \,|\, \calC)$.
We show that $\calC_0$ is contained in a subalgebra of~$\calC$ generated by 
$20\ell$-local elements of~$\calC$,
and that $\calM$ is contained a tensor product of two subalgebras of~$\calC$,
one supported on~$[a,a+20\ell]$ and the other on~$[b-20\ell, b]$.
This will complete the proof as $\calC$ is generated by $\calC_0$,
that is covered by a $20\ell$-locally generated algebra,
and the two subalgebras associated with~$\calM$,
each of which can be taken as a subset of $20\ell$-local generators.

If $x \in \calC_0 \subseteq \Mat([a+12\ell, b- 12\ell])$,
we can write it as a sum of product of single-site operators,
each of which is supported on~$[a+12\ell, b-12\ell]$.
Since~$\calA$ is invertible,
we can rewrite such a sum as another sum of products of elements of form~$y z$
where $y$ is a $\ell$-local generator of~$\calA$ (\ref{lem:csa}),
and $z$ is a $\ell$-local generator of~$\calB = \Comm(\calA, \Mat(\ZZ))$.
Applying $\tr_\calB$ of~\ref{prop:TrNormISA}, 
we express $x$ as a sum of products of
$\ell$-local elements of $\calC$.

Next, let $w \in \calM$.
Since $\calC_0$ contains $\calA_{[a+16\ell, b- 16\ell]}$,
we have by~\ref{lem:CommutantOfRestricted} an unnormalized 
Schmidt decomposition $w = \sum_j w^L_j w^R_j$,
where $w^L_j \in \calA_{[a, a+19\ell]}$
and $w^R_j \in \calA_{[b-19\ell, b]}$ by~\ref{lem:LocallyFactorizable}.
Each of $\calD^\mathrm{left}$ and $\calD^\mathrm{right}$ cannot overlap 
with both~$w^L_j$ and $w^R_j$,
so we must have $w^L_j, w^R_j \in \calC$ for all~$j$.%
\footnote{
	We did not show that each $w^L_j$ belongs to~$\calM$.
}
\end{proof}

The next lemma illustrates what we are going to do.
The assumption of this lemma will hold always, as shown in the proof 
of~\ref{thm:BrauerGroupIsTrivialIn1D}.

\begin{lemma}\label{lem:findCSin1d}
Assume that for any given site $s \in \ZZ$
there exist two $*$-subalgebras~$\calD_s^\mathrm{right}, \calD_s^\mathrm{left} $
such that 
\begin{itemize}
\item[(i)] $\calD_s^\mathrm{left}, \calD_s^\mathrm{right} \subseteq \calA_{[s - 10\ell, s+10\ell]}$,
\item[(ii)] $\calD_s^\mathrm{right}$ commutes with $\calA_{[s-30\ell, s-3\ell]}$,
 $\calD_s^\mathrm{left}$ commutes with $\calA_{[s+3\ell, s+30\ell]}$,
\item[(iii)] $\Supp \Cent \Comm(\calD_s^\mathrm{right} \,|\, \calA_{[s - 30\ell, s + 3\ell]} ) \subseteq [s-30\ell, s-7\ell]$, and
\item[]
 $\Supp \Cent \Comm(\calD_s^\mathrm{left} \,|\, \calA_{[s - 3\ell, s + 30\ell]} ) \subseteq [s + 7\ell, s+30\ell]$.
\end{itemize}
Then, $\calA$ is generated by mutually commuting 
central simple $*$-subalgebras $\calA_i \subset \calA$ indexed by $i \in \ZZ$
such that $\Supp \calA_i \subseteq [50\ell i - 49 \ell , 50 \ell i + 49 \ell ]$.
\end{lemma}

The subalgebras~$\calD^\mathrm{left,right}$ in the assumption
will control the position of potential central elements of their commutants.
The annotation ``left'' and ``right'' means
that they determine the left and right end of the commutant, respectively.
As typical in this kind of proof,
the constants such as~$3,10,30,50$ are never meant to be optimal.
They are chosen large enough to avoid certain overlaps.

\begin{proof}
For $a,b \in \ZZ$ such that $b - a > 40 \ell$,
we define
\begin{align}
\calC(a,b) &= \Comm(\calD_{a}^\mathrm{left} \calD_{b}^\mathrm{right} \,|\,  \calA_{[a, b]} ) .
\end{align}
By~\ref{lem:localGenerators}, $\calA_{[a + 3\ell, b - 3\ell]}$ is $20\ell$-locally generated.
Then, (ii) implies that these local generators
commute with both $\calD_a^\mathrm{left}$ and $\calD_b^\mathrm{right}$,
so $\calC(a,b)$ contains $\calA_{[a + 3\ell, b - 3\ell]}$.
If $z$ is a central element of $\calC(a,b)$,
then by~\ref{lem:CommutantOfRestricted},
$z$ must be supported
on~$L = [a, a + 5\ell]$ union $R= [b - 5\ell, b]$.
If $z = \sum_k z^L_k z^R_k$ is a Schmidt decomposition,
then by~\ref{lem:LocallyFactorizable}, $z^L_k , z^R_k \in \calA_{[a,b]}$ for all~$k$.
By construction,
we have $0 = [\calD_{a}^\mathrm{left}, z ] = \sum_k  [\calD_{a}^\mathrm{left}, z^L_k] z^R_k$,
where each commutator must vanish because $z^R_k$ are linearly independent,
so $z^L_k \in \Comm(\calD_{a}^\mathrm{left} \,|\, \calA_{[a, a + 15\ell]} ) 
\subseteq \Comm(\calD_{a}^\mathrm{left} \,|\, \calA_{[a, a + 30\ell]} )$.
Similarly, since $\calC(a,b)$ is $20\ell$-locally generated by~\ref{lem:localGenerators},
$z^L_k$ commutes with $\calC(a,b)$ for each~$k$.
Since $\Comm(\calD_{a}^\mathrm{left} \,|\, \calA_{[a, a + 30\ell]} ) \subseteq \calC(a,b)$,
it follows that~$z^L_k \in \Cent \Comm(\calD_{a}^\mathrm{left} \,|\, \calA_{[a,a + 30\ell]} )$.
By~(iii) we have $\Supp(z^L_k) \subseteq [a + 7\ell, a + 30\ell] \cap [a, a + 5 \ell] = \emptyset$,
so~$z^L_k$ is a scalar for all~$k$.
A symmetric argument
shows that $z^R_k$ is also a scalar for all~$k$,
and therefore $\calC(a,b)$ is central.

Next, we define
\begin{align}
\calA_{2j} &= \calC( 100 \ell j - 30\ell, 100 \ell j + 30 \ell )\\
\calA_{2j+1} &= \Comm( \calA_{2j} \calA_{2j+2} \,|\, \calC(100 \ell j - 30 \ell, 100 \ell (j+1) + 30 \ell ) )\nonumber
\end{align}
Since $\calC(b,a)$ with $b - a > 40\ell$
is central and finite dimensional,
it follows that $\calA_{2j},\calA_{2j+1}$ are all central.
By construction, two consecutive subalgebras $\calA_{i}$ and $\calA_{i+1}$ 
are mutually commuting for all~$i$.
If we show the claim on the support of $\calA_i$,
the mutual commutativity will immediately follow.

Let us examine the support of~$\calA_{2j+1}$.
Since the even-indexed $\calA_{2j}$ 
contains $\calA_{[100 \ell j - 20\ell, 100 \ell j + 20\ell]}$,
the odd index~$\calA_{2j+1}$ is supported on three disjoint intervals by~\ref{lem:CommutantOfRestricted}:
$L=[100 \ell j - 30\ell, 100\ell j - 18\ell]$,
$M=[100 \ell j + 18 \ell, 100 \ell (j+1) - 18 \ell]$,
and 
$R=[100 \ell (j+1) + 18 \ell,  100 \ell (j+1) + 30 \ell]$.
Let $x \in \calA_{2j+1}$.
Applying~\ref{lem:LocallyFactorizable} to $(L\cup R) \cup M$,
we have a Schmidt decomposition~$x = \sum_k x^{LR}_k x^M_k$ (unnormalized)
where $x^{LR}_k \in \calA_{[100\ell j - 30 \ell, 100 \ell j + 30\ell]} \calA_{[ 100 \ell (j+1) - 30 \ell, 100 \ell (j+1) + 30 \ell]}$.
Considering the commutation relation with $\calD^\mathrm{left}_{100\ell j-30\ell}$ and $\calD^\mathrm{right}_{100 \ell (j+1) + 30\ell}$,
we see that $x^{LR}_k \in \calA_{2j} \calA_{2j+2}$ for all~$k$.
But $x$ commutes with~$\calA_{2j} \calA_{2j+2}$,
so $x$ must commute with all the $20\ell$-local generators of~$\calA_{2j} \calA_{2j+2}$ (\ref{lem:localGenerators}),
each of which can overlap with either $L\cup R$ or $M$, but not both.
This implies that $x^{LR}_k$ commutes with every generator of~$\calA_{2j} \calA_{2j+2}$.
Since $\calA_{2j}\calA_{2j+2}$ is central, all $x^{LR}_k$ are scalars.
Therefore, $\calA_{2j+1}$ is supported on~$M$, 
completing the proof for the claim on the support of~$\calA_{2j+1}$.
The support of~$\calA_{2j}$ satisfies the claim by construction.

It remains to show that all $\calA_i$'s generate the full~$\calA$.
It suffices to check that every $\ell$-local generator~$g$ of~$\calA$ (\ref{lem:csa})
is contained in the algebra generated by~$\calA_i$'s.
Since $g$ is $\ell$-local,
we must have that $g \in \calA_{[100\ell j - 10\ell, 100 \ell (j+1) + 10\ell]}$
for some~$j$.
In particular, $g \in \calC(100\ell j - 30\ell, 100 \ell j + 30 \ell ) = \calA_{2j}\calA_{2j+1} \calA_{2j+2}$.
This completes the proof.
\end{proof}

\begin{theorem}\label{thm:BrauerGroupIsTrivialIn1D}
The Brauer group of invertible subalgebras of~$\Mat(\ZZ)$ is zero.
\end{theorem}

\begin{proof}
We will confirm that the assumption of~\ref{lem:findCSin1d} is always true.
Then, \ref{lem:findCSin1d} gives a decomposition of~$\calA$
into mutually commuting central simple $*$-subalgebras~$\calA_i$.
Defining a local operator algebra~$\Mat(\ZZ,q)$ with a local dimension assignment~$q$
by $q(s) = 1$ for all~$s \in \ZZ$ except $q(s = 50\ell i) = \dim \calA_i$ for $i \in \ZZ$,
we have the Brauer triviality of~$\calA$.
We will only construct $\calD_s^\mathrm{right}$;
the construction for~$\calD_s^\mathrm{left}$ will be completely parallel.
The reference to a site~$s$ will only complicate the notation,
so we work near a convenient point~$0 \in \ZZ$.

Let $\pi = \pi^\dag \in \calA_{[0,5\ell]}$ be a nonzero minimal projector.%
\footnote{
	Here we follow an idea in the proof of~\cite[Thm.~3.6]{FreedmanHastings2019QCA}.
}
That is, if $\tau = \tau^\dag \neq \pi$ 
is any other projector such that $\pi - \tau$ is positive semidefinite,
then $\tau = 0$.
Such a minimal projector exists because $\calA_{[0,5\ell]}$ is finite dimensional.
Define for any integer $j > 0$
\begin{align}
	\calX_j = \Comm(\pi \,|\, \calA_{[0,j \ell]}).
\end{align}
Clearly, $\calX_j \subseteq \calX_{j+1}$ for all~$j$.

For $j \ge 50$, 
consider the decomposition~$\calA_{[0,j\ell]} = \bigoplus_{\mu} \xi_\mu \calM_\mu$
where $\xi_\mu$ are minimal central projectors and $\calM_\mu$ are central simple.
The projector $\pi$ is represented in this decomposition
as $\pi = \sum_\mu \xi_\mu \pi$ where $\xi_\mu \pi \in \calM_\mu$ is a projector.
Considering the commutant of~$\pi$ in each simple subalgebra,
we see that the center of~$\calX_j$
is in the $\CC$-linear span of~$\pi \xi_\mu, \xi_\mu$.
But, by~\ref{lem:CommutantOfRestricted},
any element $\xi \in \Cent(\calA_{[0,j \ell]})$ 
is supported on~$L = [0,2\ell]$ union $R = [j\ell - 2\ell, j \ell]$.
Every Schmidt component of~$\xi$ (\ref{lem:LocallyFactorizable}) 
on either~$L$ or~$R$ belongs to $\calA_{[0, j\ell]}$.
Since $\calA_{[0,j\ell]}$ is $20\ell$-locally generated (\ref{lem:localGenerators}),
we see that each Schmidt component must be central by itself.
It follows that $\Cent( \calA_{[0,j\ell]} ) = \langle \xi_\nu^L \xi_\sigma^R \rangle$ 
is a tensor product of two commutative algebras 
generated by nonzero (possibly nonminimal) central projectors $\xi_\nu^L$
and $\xi_\sigma^R$.
Since $\pi \xi_\nu^L \preceq \pi$ and $\xi_\nu^L \in \calA_{[0,5\ell]}$,
we must have $\pi \xi_\nu^L = \pi$ or $\pi \xi_\nu^L = 0$ for the minimality of~$\pi$
for all~$\nu$.
Therefore, the center of~$\calX_j$ is generated by~$\pi$ and $\xi_\sigma^R$
where all $\xi_\sigma^R$ are supported on the right 
end~$[j\ell - 2\ell , j\ell]$.

If $\pi \xi'$ is a central projector of~$\pi \calX_{53}$,%
\footnote{
	$\pi \calX_j$ is a unital $*$-algebra on its own,
	but the inclusion $\pi \calX_j \to \calA_{[0,j \ell]}$ 
	is not a unital $*$-homomorphism;
	the unit~$\pi$ of~$\pi \calX_{50}$ is mapped to $\pi \neq \one$.
}
then $\tr(\pi \xi_\sigma^R  \pi \xi') \tr(\pi) = \tr(\pi)^2 \tr(\xi_\sigma^R) \tr(\xi') 
= \tr(\pi \xi_\sigma^R) \tr(\pi \xi')$,
which is the condition to apply~\ref{lem:CentralClosure0} to~$\pi \calX_{50} \subseteq \pi \calX_{53}$.
We obtain a central $*$-algebra $\calC^\pi$
such that $\pi \calX_{50} \subseteq \calC^\pi \subseteq \pi \calX_{53}$
where the superscript~$\pi$ is 
to note that the unit~$\pi$ of~$\calC^\pi$ 
is not equal to the unit~$\one \in \calA$.
Define
\begin{align}
 \calD^\pi &= \Comm( \calC^\pi \,|\, \pi\calX_{53} ), \\
 \calD &= \big\langle x^R_i ~|~ \sum_i x^L_i x^R_i \text{ is a Schmidt decomposition of }x \in \calD^\pi \text{ w.r.t. } \calA_{[0,20\ell]} \calA_{[40\ell,53\ell]} \big\rangle . \nonumber
\end{align}
Let us first check that $\calD$ is well defined.
Suppose $x = \pi x = x \pi \in \calD^\pi$.
If $t \in \Supp(x) \cap [20\ell, 40\ell]$,
then the VS property of~$\calA$ gives $w \in \calA_{[15\ell, 45\ell]}$
such that $[w,x] \neq 0$.
Then,~$\pi w = w \pi \in \pi \calA_{[15\ell, 45\ell]} \subseteq \pi \calX_{50} \subseteq \calC^\pi$
and $[\pi w, x] = \pi w x - x \pi w = w x - x w \neq 0$,
which is a contradiction.
So, $x$ is supported on the disjoint union of intervals
that appear in the definition of~$\calD$,
and thus $\calD$ is well defined.
We declare $\calD_{50\ell}^\mathrm{right} = \calD$ and claim that 
$\calD_{50\ell}^\mathrm{right}$ satisfies all the properties
in the assumption of~\ref{lem:findCSin1d}.

First, $\calD$ is supported on~$[40\ell, 53\ell] 
\subseteq [50\ell - 10\ell, 50 \ell + 10\ell]$.
This is~(i).
Second, we have to show that $\calA_{[20\ell, 47\ell]}$ commutes with~$\calD$.
If $g \in \calA_{[20\ell, 47\ell]}$,
then $\pi g = g \pi \in \pi \calX_{50} \subseteq \calC^\pi$,
so $\pi g$ commutes with~$\calD^\pi$.
If $\sum_i x_i^{\le 20} x^{\ge 40}_i$ is a Schmidt decomposed element of~$\calD^\pi$,
we have $\sum_i x_i^{\le 20} \pi [g,x^{\ge 40}_i] = 0$ where $x_i^{\le 20} \pi = x_i^{\le 20}$ for any~$i$,
so $[g,x^{\ge 40}_i] =0$.
This means that $g$ commutes with~$\calD$,
implying~(ii).
Finally, we have to show~(iii) that 
$\Supp \Cent \Comm(\calD \,|\, \calA_{[20\ell,53\ell]} ) \subseteq [20\ell, 43\ell]$.
Since $\Comm(\calD \,|\, \calA_{[20\ell,53\ell]} )$ 
contains $\calA_{[20\ell, 47\ell]}$ by~(ii),
any central element~$y \in \Cent \Comm(\calD \,|\, \calA_{[20\ell,53\ell]} )$ 
is supported on $L = [20\ell, 22\ell]$ union $R= [45\ell,53\ell]$ by~\ref{lem:CommutantOfRestricted}.
Let $y = \sum_i y^L_i y^R_i$ be an unnormalized Schmidt decomposition.
By~\ref{lem:LocallyFactorizable} we know $y^R_i \in \calA_{[45\ell,53\ell]}$ for all~$i$.
We claim that $y^R_i \in \CC\one$, which completes the proof of~(iii) and hence the theorem.

Since~$y$ belongs to $\Comm(\calD | \calA_{[20\ell,53\ell]})$
and $\calD$ is supported on~$[40\ell,53\ell]$,
we see that $y^R_i$ must commute with~$\calD$ for all~$i$,
so $y^R_i \in \Comm(\calD \,|\, \calA_{[20\ell,53\ell]} )$.
Hence,
\begin{align}
	\pi y^R_i \in \Comm(\calD^\pi, \pi \calX_{53}) = \calC^\pi \label{eq:piyR}
\end{align}
where the equality is because in any finite dimensional $*$-algebra ($\pi \calX_{53}$)
the bicommutant of a central $*$-subalgebra ($\calC^\pi$) is itself.
Let us show that $\calC^\pi$ admits a generating set of form~$\{ \pi h\}$
where $h \in \calA_{[0,53\ell]}$ are $20\ell$-local;
the argument is the same as in the proof of~\ref{lem:localGenerators}.
Note that $\calX_{50}$ is $20\ell$-locally generated by~\ref{lem:localGenerators}.
Thus, $\pi \calX_{50}$ is generated by $\pi$ times $20\ell$-local operators.
By~\ref{lem:CentralClosure0} we find a central $*$-algebra $\calE^\pi$ such that
$\pi \calX_{47} \subseteq \calE^\pi \subseteq \pi \calX_{50} \subseteq \calC^\pi$.
This $\calE^\pi$ is contained in an algebra generated by $20\ell$-local operators.
$\Comm(\calE^\pi \,|\, \calC^\pi)$ is supported (modulo~$\pi$) 
near the ends of the interval, which factorizes by~\ref{lem:LocallyFactorizable},
and we can take each tensor factor as generators.
Now, every $20\ell$-local generator~$\pi h$ of~$\calC^\pi$
commutes with $\pi y^R_i$ trivially if $h$ is supported on~$[0,43\ell]$.
If $h$ is supported on $[20\ell, 53\ell]$,
then by construction $h$ belongs to $\Comm(\calD \,|\, \calA_{[20\ell,53\ell]} )$.
Such $h$ commutes with~$y \in \Cent \Comm(\calD \,|\, \calA_{[20\ell,53\ell]} )$ 
and hence also with~$y^R_i$.
Therefore, $\pi y^R_i$ commutes with $\calC^\pi$.
Since~$\pi y^R_i \in \calC^\pi$ by~\eqref{eq:piyR},
we have $\pi y^R_i \in \Comm(\calC^\pi, \calC^\pi) = \CC \pi$.
Since $\pi$ and $y^R_i$ are far apart in support,
we conclude that $y^R_i \in \CC \one$.
\end{proof}

\nocite{apsrev42Control}
\bibliographystyle{apsrev4-2}
\bibliography{../refs}
\end{document}